\documentclass[aps,prx,twocolumn,
groupedaddress,longbibliography]{revtex4-2}

\usepackage{dcolumn}
\usepackage{url}
\expandafter\let\csname equation*\endcsname\relax
\expandafter\let\csname endequation*\endcsname\relax
\usepackage{amsmath,amssymb}
\usepackage{amsthm}
\usepackage{graphicx,bm,color,hyperref}
\usepackage{mathrsfs}
\usepackage{dsfont}
\usepackage{mathtools}
\usepackage{bbm,bm}
\usepackage{mathdots}
\usepackage{xcolor}
\usepackage{tabularx}
\usepackage{booktabs}
\usepackage{hyperref}
\usepackage{amssymb}
\usepackage{amsfonts}
\usepackage{changes}
\usepackage{subfigure}
\usepackage{xcolor}
\usepackage{float}
\usepackage{cleveref}
\usepackage{enumitem}
\usepackage{braket}
\usepackage{subcaption}
\usepackage{caption}

\usepackage{mhchem}

\usepackage{tikz}
\usepackage{adjustbox} 

\usepackage{changes}

\usepackage{cleveref}

\crefname{figure}{Figure}{Figures}
\crefname{corollary}{Corollary}{Corollary}
\crefname{conjecture}{Conjecture}{Conjectures}
\crefname{section}{Section}{Sections}
\crefname{appendix}{Appendix}{Appendixes}
\crefname{observation}{Observation}{Observation}
\crefname{remark}{Remark}{Remark}
\crefname{example}{Example}{Examples}
\crefname{equation}{Eq.}{Eqs.}
\crefname{table}{Table}{Tables}

\newtheorem{proposition}{Proposition}

\begin{document}

\title{ $\chi$-Colorable Graph States: Closed-Form Expressions and Quantum Orthogonal Arrays}

 \author{Konstantinos-Rafail Revis}
 \author{Hrachya Zakaryan}
 \author{Zahra Raissi} 
\affiliation{
Department of Computer Science, Paderborn University, Warburger Str. 100, 33098, Paderborn, Germany \\
Institute for Photonic Quantum Systems (PhoQS), Paderborn University, Warburger Str. 100, 33098 Paderborn, Germany
}

\begin{abstract}

Graph states are a fundamental class of multipartite entangled quantum states with wide-ranging applications in quantum information and computation. In this work, we develop a systematic framework for constructing and analyzing $\chi$-colorable graph states, deriving explicit closed-form expressions for arbitrary $\chi$. For two- and a broad family of three-colorable graph states, the representations obtained using only local operations, require a minimal number of terms in the Z-eigenbasis. We prove that every two-colorable graph state is local Clifford (LC) equivalent to a state expressible as a summation of rows of an orthogonal array (OA), providing a structured approach to writing highly entangled states. For graph states with $\chi > 2$, we show that they are LC-equivalent to quantum orthogonal arrays (QOAs), establishing a direct combinatorial connection between multipartite entanglement and structured quantum states.
Additionally, the derived closed-form expression has the minimal Schmidt measure for every two-colorable and a broad family of three-colorable graph states. Furthermore, the upper and lower bounds of the Schmidt measure are also discussed. Our results offer an efficient and practical method for systematically constructing graph states, optimizing their representation in quantum circuits, and identifying structured forms of multipartite entanglement. This framework also connects graph states to $k$-uniform and absolutely maximally entangled (AME) states, motivating further exploration of the structure of entangled states and their applications in quantum networks, quantum error correction, and measurement based quantum computing.

\end{abstract} 
\maketitle

\section{Introduction}

Multipartite entanglement is a fundamental resource in quantum information science \cite{Horodecki_2009}, playing a central role in quantum computing, quantum error correction \cite{intro-QEC-1,Scott-2004,intro-QEC-3, QECCgottesman2009introductionquantumerrorcorrection}, quantum networks \cite{intro-QN-1}, and quantum cryptography \cite{intro-secret-sharing}. Among the many families of entangled quantum states, graph states provide a structured way to encode entanglement, making them particularly useful for quantum error correction and measurement-based quantum computation \cite{gstates_review, Eisert, mb1Raussendorf_2002,Raissi_2020,Scott-2004, Raissi_2018}. Given their connection to stabilizer codes \cite{eisert_schlingemann2001}, it is natural to explore whether the structural properties of graphs can reveal deeper insights into the computational power, entanglement properties, and classification of graph states.

One such structural property is the {\it chromatic number}, which describes the minimum number of colors required to assign to the vertices of a graph such that no two adjacent vertices share the same color \cite{colorabilitydef}. The chromatic number is typically referred to as \textit{colorability}. While colorability is a well-studied problem in classical graph theory, its role in the entanglement structure of graph states remains largely unexplored. Previous works have suggested that colorability places upper and lower bounds on the Schmidt measure of graph states \cite{Eisert} and that almost all two colorable graph states have maximal Schmidt measure \cite{Severini_2006}.

In this work, we develop a systematic framework for deriving closed-form expressions for two-colorable graph states. We begin by obtaining explicit representations for two-colorable graph states, proving that they are LC-equivalent to a summation of rows of an orthogonal array (OA) \cite{OAdefhedayat1999orthogonal,QOA-def1}. This insight provides an efficient way to construct and manipulate highly entangled graph states, revealing a deep connection between the structure of graph states and combinatorial designs. Furthermore, the derived closed-form expressions achieve the lower bound of Schmidt measure given in \cite{Eisert} and thus they admit the minimal number of terms in the $Z$-eigenbasis.

We then extend this framework to three-colorable graph states, proving that they are LC-equivalent to quantum orthogonal arrays (QOAs) \cite{QOA-def1, QOA-def2, zahra_qoa,QOA-def3, QOA-kUNI-1}, and therefore can be useful in the context of their applications \cite{QOA-applications-1, QOA-applications-2}. This result highlights an intrinsic combinatorial structure in multipartite entangled states and provides an intuitive, efficient method for writing highly entangled three-colorable states directly from their graph structure. For the generic case, lower bounds on the Schmidt measure are obtained. Additionally, we analyze a special class of three-colorable graph states that include $k$-uniform and absolutely maximally entangled (AME) states \cite{Burchardt_2020}, showing how additional graph connections modify entanglement properties. For this class of three-colorable states, the Schmidt measure is explicitly determined.

To further explore the implications of colorability, we investigate the relationship between two- and three-colorable graph states under local unitary (LU) and LC operations. While LC operations can alter graph connectivity and consequently the colorability, because of local complementation \cite{BOUCHET199375, Van_den_Nest_2004, Van_den_Nest_2004_alg}, we discuss how we can determine the Schmidt measure of graph states with $\chi>2$. Furthermore, we identify cases where three-colorable graph states cannot be transformed back into their corresponding two-colorable forms using any invertible local operations \cite{zahra_adam_kuni}, highlighting fundamental constraints imposed by graph connectivity on multipartite entanglement.

Finally, we extend our framework to arbitrary $\chi$-colorable graph states, providing closed-form expressions and examining their entanglement structure by obtaining bounds on Schmidt measure. Our approach demonstrates that higher-colorable graph states exhibit entanglement patterns that relate to QOAs, offering insights into multipartite entanglement characterization. While the computational complexity of colorability remains a challenge, our results suggest that it provides a structured way to analyze entangled states, with potential applications in quantum networks and scalable quantum computing architectures.

The remainder of this paper is organized as follows. In Section \ref{label:Basics}, the basics of qudits and graph states are presented. In Section \ref{section:2-col}, the closed-form expression and the implications of the result are presented. In Section \ref{section:3col-gen}, the closed-form expression and the corresponding impacts for the three-colorable graph states are displayed.
Then, in Section \ref{section:3col-special}, we discuss a specialized set of three-colorable graph states designed to describe $k$-uniform states. In Section \ref{section:LT}, the usage of local complementation and the connection of the given framework to the $k$-uniform states is made apparent, indicating the impact of the previously presented results. In Section \ref{section:x-colorable}, we generalize our results to arbitrary $\chi$ colorability, further expanding the scope of our approach. 
Finally, we conclude in Section \ref{section:conclusion} by discussing the potential implications of our work and outlining directions for future research.

\section{Basics of qudit graph states}\label{label:Basics}
Graph states are highly entangled quantum states associated with a mathematical graph $(V, E)$, where vertices $V$ represent qudits, and edges $E$ correspond to entangling operations between them \cite{eisert_g1,eisert_g2}. These states serve as a central resource in quantum information protocols, including measurement-based quantum computation. Here, we introduce the essential formalism for qudit graph states, setting the stage for an exploration of their colorability in the following sections. Two vertices $i,j\in V$, with the total number of vertices denoted as $n$, are called adjacent if they are the endpoints of an edge. 
The adjacency matrix $\Gamma$ is a symmetric $n \times n$ matrix defined over a finite field $ \mathbb{F}_d$, where each nonzero element $\Gamma_{ij}=\theta_{ij} \in \mathbb{F}_d$ specifies the entanglement weight between qudits $i$ and $j$. Therefore \cite{Eisert}:
\begin{equation}
\Gamma_{ij}= 
\begin{cases} 
\theta_{ij} & \text{if } i,j \in E  \\
0 & \text{otherwise } 
\end{cases}\ .
\end{equation}

Mentioning the neighborhood $N_i\subset V$ of a graph $G$ is useful, which is the set of vertices for which $\{i,j\}\in E$. Let us briefly review the concept of qudits, which are systems with $d$ levels, such that the Hilbert space of each particle is $\mathbb{C}^d$. The action of the Pauli operators $Z$ and $X$ on the eigenstates of $Z$ is defined as follows \cite{quditsLCbahramgiri2007graphstatesactionlocal,zahra_qudits}.
\begin{equation}
\label{eq:defxz} 
    X^a\ket{i} = \ket{i+a} \quad \text{and} \quad Z^a\ket{i} = \omega^{ia}\ket{i},
\end{equation}
where $\omega = e^{i 2\pi/d}$ the $d$th root of unity, $a\in\mathbb{F}_d$, and $X^d=Z^d=\mathbb{I}_d$. All additions are performed modulo $d$. 
A critical operation for qudits is the Hadamard gate, whose action is defined as:
\begin{equation}
    \label{eq:hgate}
    H\ket{i} = \frac{1}{\sqrt{d}}\sum_{l=0}^{d-1}\omega^{il}\ket{l},\quad \text{and} \quad H^{\dagger}\ket{i} = \frac{1}{\sqrt{d}}\sum_{l=0}^{d-1}\omega^{-il}\ket{l}.
\end{equation}
Analogous to the qubit case, the controlled-$Z$ ($CZ$)  gate is defined for qudits as: 
\begin{equation}
\label{eq:cz}
    C_1Z_2^{\beta}\ket{i}_1\ket{j}_2 = \omega^{ij\beta}\ket{i}_1\ket{j}_2,
\end{equation}
where $\beta\in\mathbb{F}_d$.  The subscripts in $CZ$ indicate the control (first qudit) and target (second qudit), but since $CZ$ is symmetric under the exchange of qudits, its effect remains unchanged if the indices are swapped. In contrast, for $CX$, the order of the indices is crucial: the first index corresponds to the control qudit, while the second index receives the modular addition transformation,
\begin{equation}
\label{eq:cx}
C_1X_2^{\gamma}\ket{i}_1\ket{j}_2 = \ket{i}_1\ket{i+j+\gamma}_2,
\end{equation}
where $\gamma\in\mathbb{F}_d$. In this case, the order of the indices is significant.

Finally, some useful relations relevant to the upcoming results are presented in \ref{appen:useful}. For a graph $(V,E)$, where $V$ and $E$ denote the sets of vertices ($|V|=n$) a graph state is defined as \cite{Eisert}: 
\begin{equation}
\label{eq:gs-definition}
      \ket{\psi_{\text{GS}}}=\left(\displaystyle\prod_{\{k,j\}\in E}C_{k}Z_{j}^{\Gamma_{kj}}\right)\ket{+}^{\otimes n},\quad \text{where} \ \ \ket{+} = H\ket{0}.
\end{equation}
It should be noted that the normalization factor here is not omitted but, in general, is omitted unless it is essential to the discussion. Here, the subscript GS indicates that the state $\ket{\psi_{\text{GS}}}$ corresponds to the graph representation.

\section{Two-colorable Graph States }\label{section:2-col}

Let $\ket{\psi_{\text{2-color}}}$ be a two-colorable graph state. 
A graph is two-colorable if its chromatic number is $\chi=2$, meaning that its vertices can be partitioned into two disjoint sets, which we label as $B$ (blue vertices) and $R$ (red vertices). The total number of red vertices is $n_R$, and the total number of blue vertices is $n_B$, so the total number of vertices in the graph is $n=n_R + n_B$.  Throughout this section, the notation $r \in R$ refers to a particle in the set of red vertices, and $b \in B$ refers to a particle in the set of blue vertices.

We assume that the  $CZ$ operations in the graph state construction occur between red vertices as controls and blue vertices as targets. The two-colorable graph state $\ket{\psi_{\text{2-color}}}$, is then expressed as:
\begin{align}
\label{eq:2colgsdef}
    \ket{\psi_{\text{2-color}}}=\displaystyle\prod_{r\in R, b \in B}C_rZ_{b}^{\Gamma_{rb}}\ket{+}^{\otimes R}\ket{+}^{\otimes B} \ ,
\end{align}
This expression highlights the structural simplicity of two-colorable graph states, as the adjacency matrix of any two-colorable graph state takes a specific block form
\begin{equation}
\label{eq:gamma-2-col}
\Gamma_{2\text{-color}} = \Gamma_{rb} =  \left(\begin{array}{c|c} \smash{\overbrace{O}^{\text{B}}} & \smash{\overbrace{A_{RB}}^{\text{R}}} \\ \hline A_{RB}^T & O \end{array}\right)\ ,
\end{equation}
where the block $A_{RB}$ is the $n_R\times n_B$ block containing the weights of edges connecting red and blue vertices. The diagonal blocks are zero matrices because vertices of the same color are not connected in two-colorable graphs.

Two-colorable graph states have significant applications in quantum error correction, measurement-based quantum computing, and entanglement quantification \cite{eisert_5,zander2024benchmark}. However, the explicit construction of these states using conventional methods is computationally intensive, particularly for large graphs. A closed-form representation directly derived from the graph's adjacency matrix simplifies the state construction and enables practical implementation of these states in quantum protocols.

In the following, we derive a closed-form representation for two-colorable graph states and explore its implications for their construction, analysis, and applications.

\begin{proposition}\label{thm-2col}
    Assume a two-colorable graph state given by equation (\ref{eq:2colgsdef}), with adjacency matrix (\ref{eq:gamma-2-col}) and without loss of generality $n_B\geq n_R$. The state satisfies the following closed-form expression: 
    \begin{equation}
    \label{eq:closeform2Col1}
    H^{\dagger\otimes B}\ket{\psi_{\text{2-color}}}=\displaystyle\sum_{\vec{i}=0}^{d-1}\ket{\vec{i}\mathcal{G}},
\end{equation}
where $\vec{i}=(i_1,i_2,...,i_{n_R})$ is a row vector, and $\mathcal{G}=\begin{bmatrix} \mathbb{I}_{n_R} & | & A_{RB}\end{bmatrix}$, with $A_{RB}$ being the top-right block of the adjacency matrix.
\end{proposition}

\begin{proof}
 To prove the closed-form expression for $\ket{\psi_{\text{2-color}}}$, we assume without loss of generality that $n_B\geq n_R$. By applying $H^{\dagger}$ to every blue particle, we start with:
    \begin{equation*}
    H^{\dagger\otimes B}\ket{\psi_{\text{2-color}}} = \prod_{r \in R, b \in B} H^{\dagger\otimes B} C_r Z_b^{\Gamma_{rb}} \ket{+}^{\otimes R} \ket{+}^{\otimes B}.
    \end{equation*} 
Using the commutation relation $H^\dagger Z H=X$ for Pauli operators, the $H^\dagger$ gate commutes with the $CZ$ operation as, (see also (\ref{eq:hdagczcomm})),
 \begin{equation*} 
H^{\dagger} C_r Z_b^{\Gamma_{rb}} = C_r X_b^{\Gamma_{rb}} H^{\dagger}\ .
 \end{equation*} 
 Substituting this commutation relation, we obtain: 
    \begin{equation*}
        H^{\dagger\otimes B}\ket{\psi_{\text{2-color}}} = \prod_{r \in R, b \in B} C_r X_b^{\Gamma_{rb}} \ket{+}^{\otimes R} H^{\dagger\otimes B} H^{\otimes B} \ket{0}^{\otimes B}.
    \end{equation*}
This simplifies to:
\begin{equation*}
 H^{\dagger\otimes B}\ket{\psi_{\text{2-color}}}= \prod_{r \in R, b \in B} C_r X_b^{\Gamma_{rb}} \ket{+}^{\otimes R} \ket{0}^{\otimes B}. 
\end{equation*}
Recalling that $\ket{+} = H\ket{0}$, the state becomes:
\begin{align*}
     H^{\dagger\otimes B}\ket{\psi_{\text{2-color}}} &= 
      \sum_{\vec{i}=0}^{d-1}  \prod_{r \in R, b \in B} & C_r X_b^{\Gamma_{rb}} \ket{i_1\ldots i_{n_R}} \ket{0}^{\otimes B}.
\end{align*}   

Applying the $CX$ operation,  the blue particles are updated as: 
\begin{align*}
     &H^{\dagger\otimes B}\ket{\psi_{\text{2-color}}}= 
     & \sum_{\vec{i}=0}^{d-1} \prod_{r \in R, b \in B} \ket{i_1\ldots i_{n_R}} X_{b}^{ \Gamma_{rb}\cdot i_r} \ket{0}^{\otimes B}.
\end{align*}

For each blue particle $b \in B$, the $X_b$ operator applies shifts determined by the sum of neighboring red particle indices:
\begin{align*}
 H^{\dagger\otimes B}\ket{\psi_{\text{2-color}}} &=
  \sum_{\vec{i}=0}^{d-1} \prod_{b \in B} & \ket{i_1\ldots i_{n_R}} X_b^{\sum_{r \in R} \Gamma_{rb}\cdot i_r} \ket{0}^{\otimes B}.    
\end{align*} 

Finally, applying the $X_b$ operators to $\ket{0}^{\otimes B}$ yields:
\begin{equation*}
    H^{\dagger\otimes B}\ket{\psi_{\text{2-color}}}=\sum_{\vec{i}=0}^{d-1} \ket{i_1\ldots i_{n_R}} \bigotimes_{b \in B} \ket{\sum_{r \in R} \Gamma_{rb}\cdot i_r}.
\end{equation*}
The indices of the red particles form a row vector $\vec{i}=(i_1,i_2,...,i_{n_R})$ and defining the generator matrix as $\mathcal{G}=\begin{bmatrix} \mathbb{I}_{n_R} & | & A_{RB}\end{bmatrix}$, where $A_{RB}$ is the top right block of adjacency matrix $\Gamma_{2\text{-color}}$, the equation simplifies to: 
\begin{equation}
    \label{eq:closeform2Col2}
    H^{\dagger\otimes B}\ket{\psi_{\text{2-color}}}=\displaystyle\sum_{\vec{i}=0}^{d-1}\ket{\vec{i}\mathcal{G}} \ . 
\end{equation}
This concludes the proof of the closed-form expression.
\end{proof}

This result of the proposition \ref{thm-2col} provides a direct, step-by-step method for constructing the state vector of any two-colorable graph state. The procedure is as follows:
\begin{itemize}
    \item Step 1: In a two-colorable graph, identify the red and blue vertex sets ($R$ and $B$). Without loss of generality, we have assumed $n_B\geq n_R$. 
    \item Step 2: Assign free indices to the red vertices.
    \item Step 3: Compute the indices for the blue vertices as the sum of their neighboring red vertex indices.
    \item Step 4: Combine these assignments into the state expression and sum over the red vertex indices.
\end{itemize}

To demonstrate the utility of Proposition \ref{thm-2col}, let us consider two examples: a six-particle circular graph and a six-particle cluster graph figure \ref{fig:combined_ex_2col}.

\begin{figure*}[htbp]
    \centering
    \begin{minipage}[b]{0.41\textwidth}
        \centering
        \includegraphics[width=\textwidth]{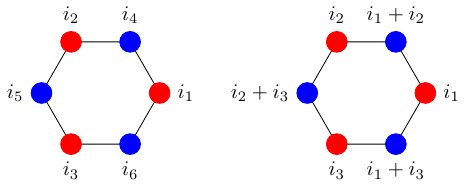} 
        \caption*{(a): Implementation on graph state corresponding to a six-particle circle.}
    \end{minipage}
    \hfill
    \begin{minipage}[b]{0.41\textwidth}
        \centering
        \includegraphics[width=\textwidth]{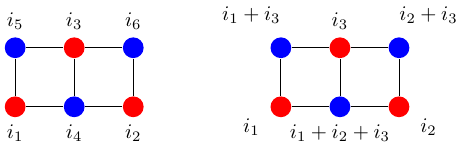} 
        \caption*{(b): Implementation on graph state corresponding to cluster state with six particles.}
    \end{minipage}
    \caption{Two examples of the implementation of the formula obtained in the proposition \ref{thm-2col}. On the left of each sub-figure, we have the indices assigned using the definition of graph states given by equation (\ref{eq:2colgsdef}), and on the right, we have the assignment of indices using the obtained result. }
    \label{fig:combined_ex_2col}
\end{figure*}

\noindent
{\bf Example 1.} Six-Particle Circular Graph: Using the graph in Figure \ref{fig:combined_ex_2col} (a), the graph state is written using the definition in equation (\ref{eq:2colgsdef}) as: 
\begin{align}
\label{eq:2col-ex1-gs}
     \ket{\Psi_{GS}} &=  \sum_{i_1,\dots,i_6=0}^{d-1} \quad \omega^{i_1 i_4} \omega^{i_4 i_2} \omega^{i_2 i_5} \omega^{i_5 i_3} \omega^{i_3 i_6} \omega^{i_6 i_1} \nonumber \\
    &\ket{i_1, i_2, i_3}\ket{i_4, i_5, i_6}.
\end{align}

In this example, the red vertices are labeled as $\{i_1,i_2,i_3\}$, and the blue vertices are $\{i_4,i_5,i_6\}$. Using Proposition \ref{thm-2col}, we assign free indices $i_1,i_2,i_3$ to the red vertices. For the blue vertices, their indices are determined as the summation of the indices of their neighboring red vertices. For instance:
\begin{align*}
i_4 &= i_1 + i_2\\
i_5 &= i_2 + i_3\\
i_6 &= i_3 + i_1 \ .
\end{align*}
The closed-form expression directly follows:
\begin{equation}
    \label{eq:2col-ex1-final}
    \ket{\psi_{\text{closed-form}}}= \sum_{i_1,i_2,i_3 = 0}^{d-1}\ket{i_1,i_2,i_3}\ket{i_1+i_2,i_2+i_3,i_3+i_1}.
\end{equation}
This example demonstrates how to use the graph's structure to implement the closed-form expression.

\noindent
{\bf Example 2.} Six-Particle Cluster Graph: For the graph in Figure \ref{fig:combined_ex_2col} (b), the graph state is expressed using equation (\ref{eq:2colgsdef}) as:
\begin{align}
\label{eq:2col-ex2-gs}
    \ket{\Psi_{GS}} &= \sum_{i_1,\cdots,i_6=0}^{d-1} 
    \omega^{i_ii_5}\omega^{i_1i_4}
    \omega^{i_3i_5}\omega^{i_3i_4}\omega^{i_3i_6}
    \omega^{i_2i_4}\omega^{i_2i_6}\nonumber \\
    &\quad \ket{i_1, i_2, i_3}\ket{i_4, i_5, i_6}.
\end{align}
Here, again the red vertices are $\{i_1,i_2,i_3\}$, and the blue vertices are$\{i_4,i_5,i_6\}$. Assigning free indices $i_1,i_2,i_3$ to the red vertices, the indices of the blue vertices are computed as:
\begin{align*}
i_4 &= i_1 + i_2 + i_3\\
i_5 &= i_1 + i_3\\
i_6 &= i_2 + i_3 \ .
\end{align*}
The closed-form expression is:
\begin{align}
\label{eq:2col-ex2-final}
& \ket{\psi_{\text{closed-form}}} =\nonumber \\ &\sum_{i_1,i_2,i_3 = 0}^{d-1} \ket{i_1,i_2,i_3} 
\ket{i_1+i_2+i_3,i_1+i_3,i_2+i_3}.
\end{align}

These examples illustrate how the closed-form expression simplifies the representation of two-colorable graph states.

Proposition \ref{thm-2col} has important implications for the Schmidt measure of two-colorable graph states. From \cite{Eisert}, any state vector $\ket{\phi}\in\mathcal{H}^{(1)}\otimes\cdots\otimes\mathcal{H}^{(n)}$ of a compound quantum system with $n$ components can be expressed as:
\begin{equation}
\label{eq:ketSM}
    \ket{\phi} = \sum_{i=1}^{\Lambda}\zeta_{i}\ket{\phi_i^{(1)}}\otimes\cdots\otimes\ket{\phi_i^{(n)}},
\end{equation}
where $\zeta_i\in\mathbb{C}$ for $i=1,\cdots,\Lambda$ and $\ket{\phi_i^{(j)}}\in \mathcal{H}^{(j)}$. Bearing this information into account, the Schmidt measure is defined as \cite{Eisert,SM-Eisert}:
\begin{equation}
    E_s(\ket{\psi})=\text{log}_2(\lambda),
\end{equation}
where $\lambda$ is the minimal number $\Lambda$ of the terms in the sum presented in (\ref{eq:ketSM}), over every linear decomposition into product states. Considering proposition 7 of \cite{Eisert}, according to which the Schmidt measure of a two-colorable graph state is bounded as follows:
\begin{equation}
\label{eq:SMbound1}
    \frac{1}{2}\text{rank}\left(\Gamma_{\text{2-color}}\right)\leq E_s(\ket{\psi_{\text{2-color}}})\leq\lfloor\frac{n_B+n_R}{2}\rfloor.
\end{equation}
This result is crucial to use, because according to proposition \ref{thm-2col} when the $H^{\dagger}$ operation is applied to every blue particle, while always assuming that $n_B\geq n_R$, the number of terms obtained in the closed-form is $d^{n_R}$. It is trivial to prove that:
\begin{equation}
    \label{eq:rank-proof1}
    \frac{1}{2}\text{rank}\left(\Gamma_{\text{2-color}}\right) = \frac{1}{2}\left(2\cdot \text{rank}\left(A_{RB}\right)\right)  = \text{rank}\left(A_{RB}\right).
\end{equation}

Finally, the minimum value of the rank of the matrix $A_{RB}$ is:
\begin{equation}
    \label{eq:rank-proof2}
    \text{min}\{\text{rank}\left( A_{RB}\right)\} = \text{min}\left(n_R,n_B\right).
\end{equation}
Therefore, using equation (\ref{eq:rank-proof2}) in equation (\ref{eq:SMbound1}), we have that for any two-colorable graph state the Schmidt measure is bounded as:
\begin{equation}
\label{eq:SMbound2}
    \text{min}\left(n_R,n_B\right)\leq E_s(\ket{\psi_{\text{2-color}}})\leq\lfloor\frac{n_B+n_R}{2}\rfloor.
\end{equation}
With proposition \ref{thm-2col}, we demonstrated how to systematically arrive at a state that has the minimal Schmidt measure. Furthermore, with proposition \ref{thm-2col} we demonstrated the strategy where the upper and the lower bound of the equation (\ref{eq:SMbound2}) can be obtained.

It must be underlined that the bounds on the Schmidt measure given in equation (\ref{eq:SMbound1}) can be applied for any local dimension $d$. This is true since it is trivial to change the base of the logarithm from $\text{log}_2$ to $\text{log}_d$ and then the factors will cancel out since they are greater than one. Therefore, it can be stated that the number of terms needed to describe any two-colorable graph state in the Z-eigenbasis, denoted as $N(\ket{\psi_{\text{2-color}}})$, is bounded as follows:
\begin{equation}
    d^{n_R}\leq N(\ket{\psi_{\text{2-color}}})\leq d^{\lfloor \frac{n_B+n_R}{2}\rfloor}, \text{where}\  n_R\leq n_B,
\end{equation}
where the lower bound is obtained by applying $H^{\dagger}$ to every blue particle and the upper bound is obtained by applying $H^{\dagger}$ to in total $\lfloor\frac{n_B+n_R}{2}\rfloor$ particles.

A final vital remark on the implications of proposition \ref{thm-2col} is that every two-colorable graph state can be written as a state that can be defined based on an orthogonal array. According to \cite{OAdefhedayat1999orthogonal}, an orthogonal array denoted as $OA(r, n,d,k)$ is a positioning formed up by $r$ rows, $n$ columns, and the entries are taking values ranging from the set $\{0,\cdots,d-1\}$. The important property here is that every subset with $k$ columns has all the combinations of symbols, which occur the same number of times through the rows. As per \cite{OAdefhedayat1999orthogonal}, an OA can be defined if a row vector $\vec{x} = (x_1,\cdots,x_d)$ and a generator matrix of the form $\mathcal{G}=\begin{bmatrix} \mathbb{I}_{d} & | & M\end{bmatrix}$ are assumed where the matrix $M$ is an $d\times p$ matrix. Then, the OA is of the form $OA = \vec{x}\cdot \mathcal{G}$ and this form is exactly what was obtained in proposition \ref{thm-2col}.

To understand the notion OA better, let us briefly discuss the $OA(4, 4, 2, 2) $ that has 4 rows, 4 columns, 2 levels per factor, and a strength of 2. This means that for any two columns chosen, all combinations of the levels 0 and 1 will appear equally across the trials. As an example, the following can be written:
\begin{equation}
\begin{array}{cccc}
0 & 0 & 0 & 0 \\
0 & 1 & 1 & 1 \\
1 & 0 & 1 & 1 \\
1 & 1 & 0 & 0 \\
\end{array}\nonumber
\end{equation}
Two OAs are equivalent if one can be transformed into the other by applying permutations or relabeling symbols within rows or columns. As it is explained in \cite{zahra_qoa}, states with this property demonstrate a high persistence in entanglement, making them great candidates for the development of multipartite quantum information protocols and in our case we demonstrated that these advantages are apparent for every two-colorable graph state.

\section{Three-Colorable graph states}\label{section:3col-gen}

Extending our analysis to three-colorable graph states ($\chi =3$), we now consider graphs where the vertex set can be partitioned into three disjoint subsets. Unlike the two-colorable case, three-colorability introduces additional structural complexity, and determining a valid 3-coloring is known to be NP-hard \cite{Karp1972,stockmeyer1973planar,lovasz1973coverings}. Let $R,G,B$ denote the sets of blue, green, and red vertices, respectively, with cardinalities $n_R, n_G$ and $n_B$. The total number of vertices in the graph is then given by $n=n_R + n_G +n_B$. In our notation, choosing a vertex $r \in R$ corresponds to selecting a red vertex, and similarly for $b\in B$ and $g \in G$.
 Therefore $\ket{\psi_{\text{3-color}}}$ is given as follows:
 \begin{widetext}
\begin{align}
\label{eq:3coldefgeneral}
    \ket{\psi_{\text{3-color}}} = \left(\prod_{r \in R, b \in B} C_r Z_{b}^{\Gamma_{rb}}\right)
     \left(\prod_{r \in R, g \in G} C_r Z_{g}^{\Gamma_{rg}}\right)
     \left(\prod_{g \in G, b \in B} C_g Z_{b}^{\Gamma_{gb}}\right)\ket{+}^{\otimes R}\ket{+}^{\otimes G}\ket{+}^{\otimes B}.
\end{align}
\end{widetext}
Partitioning the adjacency matrix into blocks clarifies the interactions between different color groups, allowing us to systematically construct the graph state. The adjacency matrix for $\chi = 3$ is denoted as:
\begin{equation}
\label{eq:gamma-3col-general}
\Gamma_{3\text{-color-general}} = \left(\begin{array}{c|c|c}
  \smash{\overbrace{0}^{\text{R}}} & \smash{\overbrace{A_{GR}}^{\text{G}}} & \smash{\overbrace{A_{BR}}^{\text{B}}} \\ \hline 
  A_{GR}^T & 0 & A_{BG} \\ \hline
  A_{BR}^T & A_{BG}^T & 0 \\ 
\end{array}\right).
\end{equation}

\begin{proposition}\label{thm-3col-general}
Let us assume a three-colorable graph state described by the adjacency matrix (\ref{eq:gamma-3col-general}), defined as in equation (\ref{eq:3coldefgeneral}), and satisfying the set conditions described at the beginning of this section. Then, any such graph state satisfies the following equation:
\begin{equation}
\label{eq:3colfinalresultv2}
    H^{\dagger\otimes B} \ket{\psi_{\text{3-color}}} =  
   \sum_{\vec{w}=0}^{d-1}
    \left( \bigotimes_{j=1}^{n_G} Z_{g_j}^{\vec{u} \cdot (\vec{A}_{GR})_{g_j}} \right)\ket{\vec{w} \cdot \mathcal{G}_{RB,GB}},
\end{equation}
where $\vec{u}=(u_1,u_2...,u_{n_R})$, $\vec{v}=(v_1,\cdots,v_{n_G})$ both row vectors, and $\vec{w}=(\vec{u},\vec{v})$ is the concatenation of $\vec{u}$ and $\vec{v}$. The index $g_j$ denotes a particle having the green color and therefore $j\in\{1,\cdots,n_G\}$. The matrix $
\mathcal{G}_{RB,GB}$ (additional connection between green and blue vertices compared to 
$\mathcal{G}_{RB}$) is a generator-like matrix with dimensions $(n_R+n_G)\times (n_R+n_G+n_B)$ defined as:
\begin{equation}
    \label{eq:G-3col-general1}
   \mathcal{G}_{RB,GB} = 
\begin{bmatrix} 
\mathbb{I}_{n_R+n_G} & | & \begin{bmatrix} A_{BR} \\ A_{BG} \end{bmatrix}
\end{bmatrix},
\end{equation}
with $A_{BR}$ and $A_{BG}$ the corresponding blocks of matrix (\ref{eq:gamma-3col-general}). Finally, the vector $(\vec{A}_{GR})_{g_j}$, denotes the vector constructed from extracting the row $g_j$ of the matrix $A_{GR}$, which corresponds to the row containing all the elements of the adjacency matrix (\ref{eq:gamma-3col-general}) that describe a specific green particle $g_j$. 

\end{proposition}

\begin{proof}
Let us start with the definition given in equation (\ref{eq:3coldefgeneral}). Recalling equation (\ref{eq:3-col-proof-1}), we can define the following row vectors $\vec{u}=(u_1,u_2...,u_{n_R})$, $\vec{v}=(v_1,\cdots,v_{n_G})$. It is useful to define the vector $\vec{w}=(\vec{u},\vec{v})$, the concatenation of $\vec{u}$ and $\vec{v}$. Let us perform the $CZ$ operations that involve the blue particles:

\begin{equation}
\begin{aligned}
    H^{\dagger \otimes B} \ket{\psi_{\text{3-color}}} = & \, H^{\dagger \otimes B} 
  \sum_{\substack{\vec{u}, \vec{v}=0}}^{d-1}
     \prod_{r \in R, g \in G} C_r Z_{g}^{\Gamma_{rg}}  \\
    & \ket{\vec{u}} \ket{\vec{v}} 
     \bigotimes_{b \in B} Z_b^{f_r(b) + f_g(b)} \ket{+}_b.
\end{aligned}
\end{equation}
where $f_r(b) = \sum_{r\in R}u_r\Gamma_{rb}$ and $f_g(b) = \sum_{g\in G}v_g\Gamma_{gb}$. We should note that, till this point, we have not performed the $H^{\dagger \otimes B}$ operation. It is clear that since it is a local operation acting only on the blue particles, it will commute via the remaining CZ operation between the green and the red particles. This will lead to:
\begin{equation}
    H^{\dagger \otimes b}Z_b^{f_r(b)+f_g(b)} = X_b^{f_r(b)+f_g(b)}H^{\dagger \otimes b},\quad \text{with} \quad b\in B.
\end{equation}
This implies that:
\begin{equation}
\label{eq:3-col-proof-1}
\begin{aligned}
    H^{\dagger \otimes B} \ket{\psi_{\text{3-color}}} = & 
    \sum_{\substack{\vec{u},\vec{v}=0}}^{d-1}
    \left(\prod_{r \in R, g \in G} C_r Z_{g}^{\Gamma_{rg}}\right) \\
    & \ket{\vec{u}} \ket{\vec{v}} \bigotimes_{b \in B} \ket{f_r(b) + f_g(b)}_b.
\end{aligned}
\end{equation}

At this point, there is a remaining CZ operation between the green and the red particles. Unfortunately, the trick to use another Hadamard gate can not be applied here. We have to understand that the primary goal is to find a way to obtain as few terms as possible in the Z-eigenbasis. In any case, we seek to reduce the number of phases in front of each bracket, but this would unavoidably increase the number of terms used in the Z-eigenbasis. 

We can define the matrix $\mathcal{G}_{RB,GB}$, which is generator-like matrix with dimensions $(n_R+n_G)\times (n_R+n_G+n_B)$, as:
\begin{equation}
    \label{eq:G-3col-general2}
   \mathcal{G}_{RB,GB} = 
\begin{bmatrix} 
\mathbb{I}_{n_R+n_G} & | & \begin{bmatrix} A_{BR} \\ A_{BG} \end{bmatrix}
\end{bmatrix},
\end{equation}
with $A_{BR}$ and $A_{BG}$ the corresponding blocks of matrix (\ref{eq:gamma-3col-general}). In this juncture, the equation (\ref{eq:3-col-proof-1}) will read as:

\begin{equation}
\label{eq:3-col-proof-2}
\begin{aligned}
    H^{\dagger \otimes B} \ket{\psi_{\text{3-color}}} =  
    \sum_{\vec{w}=0}^{d-1}
    \left( \prod_{r \in R, g \in G} C_r Z_{g}^{\Gamma_{rg}} \right)\ket{\vec{w} \cdot \mathcal{G}_{RB,GB}},
\end{aligned}
\end{equation}
where $\vec{w}=0$ denotes the range of the indices $u_1, \dots, u_{n_R}, \\v_1, \dots, v_{n_G}=0$.

The final step is to realize the power of the adjacency matrix (\ref{eq:gamma-3col-general}). Assume that the index $g_j$ denotes a particle having the green color and therefore $j\in\{1,\cdots,n_G\}$. Finally, define the vector $(\vec{A}_{GR})_{g_j}$, to denote the vector constructed from extracting the row $g_j$ of the matrix $A_{GR}$, which corresponds to the row containing all the elements of the adjacency matrix (\ref{eq:gamma-3col-general}) that describe a specific green particle $g_j$.
\begin{align}
     H^{\dagger\otimes B} \ket{\psi_{\text{3-color}}} = \nonumber 
   \sum_{\vec{w}=0}^{d-1}
    \left( \bigotimes_{j=1}^{n_G} Z_{g_j}^{\vec{u} \cdot (\vec{A}_{GR})_{g_j}} \right)\ket{\vec{w} \cdot \mathcal{G}_{RB,GB}}.
\end{align}

This concludes our proof.
\end{proof}
At this point, we turn our focus on discussing the implications of our proposition \ref{thm-3col-general}. Let us start with the fact that we are equipped with an easy way to write a general three-colorable graph state. Here a step-by-step process is presented:
\begin{itemize}
    \item Step 1: In a three-colorable graph identify the colors of the graph and assign the colors following that without loss of generality $n_R\le n_G\le n_B$.
    \item Step 2: Assign free indices to the red and green particles.
    \item Step 3: The indices of the blue particles are the summation of the neighboring indices of each blue particle weighted with the corresponding weight of the adjacency matrix.
    \item Step 4: To identify the correct exponent for the Z operations for each green particle, we need to identify the red neighbors of each green particle and multiply their free indices scaled with the adequate weight imposed by the adjacency matrix (\ref{eq:gamma-3col-general})
\end{itemize}
\begin{figure*}[htbp]
    \centering
    \includegraphics[width=0.5\linewidth]{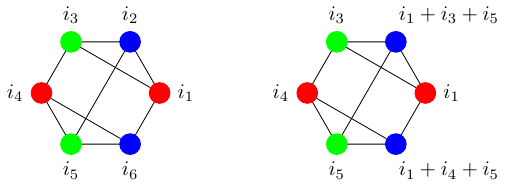}
    \caption{An example of implementation of the closed-form expression for a 6 particle three-colorable graph state.}
    \label{fig:3col-gen-ex}
\end{figure*}
The state given in the left-hand side of figure \ref{fig:3col-gen-ex} can be written as:
\begin{equation}
\begin{aligned}
    \ket{\psi_{\text{GS}}} = & \, 
    \sum_{i_1, \ldots, i_6=0}^{d-1} 
    \omega^{i_1 i_2 + i_2 i_3 + i_3 i_4 + i_4 i_5 + i_5 i_6 + i_1 i_3 + i_4 i_6 + i_2 i_5} \\
    & \ket{i_1, i_2, i_3, i_4, i_5, i_6}.
\end{aligned}
\end{equation}

Applying our methodology, the indices for the red and green particles are kept as they were before, but the indices for the red particles are given with respect to the red and green particles, as the right-hand side of figure \ref{fig:3col-gen-ex} indicates. With these steps, we have implemented steps 1 to 3. The determination of the correct exponents for the Z operators is easy, since one has to simply identify the green particles that have a red connection. Then multiply the indices with the corresponding weight coming from the adjacency matrix, which in this case is one. So, the closed-form expression of this example is:
\begin{equation}
\begin{aligned}
    \ket{\psi_{\text{3-color}}} = & \, 
    \sum_{i_1, i_3, i_4, i_5 = 0}^{d-1} 
    Z_{g_{i_3}}^{i_1 i_3} Z_{g_{i_5}}^{i_4 i_5} \\
    & \ket{i_1, i_1 + i_3 + i_5, i_3, i_4, i_5, i_1 + i_4 + i_5}.
\end{aligned}
\end{equation}

To avoid confusion, the indices $g_{i_3}$ and $g_{i_5}$ are used to denote the Z operations acting on the green particles that in figure \ref{fig:3col-gen-ex}, are denoted with indices $i_3$ and $i_5$. It is important to note that there is a difference between the indexing to implement the closed-form expression and the indices to name and group the particles according to their color.
At this juncture, it has to be pointed out that the above example was not chosen randomly. According to \cite{zahra_qoa}, this state is also an absolutely maximally entangled state of six particles with local dimension $d$ being a prime number. We denote these states as $AME(6,d)$. In the next section, \ref{section:3col-special} we are going to tailor the formalism presented here to better represent graph states that are highly entangled.

Our analysis shows that every three-colorable graph state can be rewritten in the form of a QOA, reinforcing the deep connection between graph states and highly entangled quantum states \cite{QOA-def1,QOA-def2,QOA-def3,zahra_qoa}. QOAs are combinatorial designs that contain rows composed of pure multipartite quantum states and exhibit high entanglement properties. Bearing in mind the definition of the classical OAs, the QOAs are constructed assuming specific parameters for the OA \cite{QOA-construction-1}. As previously mentioned, QOAs display high entanglement properties and therefore they can be used to generate $k$-uniform states and absolutely maximally entangled states \cite{QOA-def2}. Additionally, they are rather useful in decoupling schemes in quantum networks \cite{QOA-applications-1,QOA-applications-2}. Various QOA setups have been studied in works like \cite{zahra_qoa,QOA-kUNI-1}. 

In general, proposition \ref{thm-3col-general} implies that every three-colorable graph state can be written as a QOA. Initially, it is trivial to understand that if the correct vector $\vec{w}$ and the correct $\mathcal{G}_{RB,GB}$ matrix, following the notation from proposition \ref{thm-3col-general}, is used then, the object $\vec{w}\cdot \mathcal{G}_{RB,GB}$ represents an OA. The reason why any three-colorable graph, is in general, a QOA relies on the additional Z phases that we were not able to eliminate. In the scope of graph states, it is easy to understand if we recall that every three-colorable graph can become two-colorable if the appropriate number of edges is deleted. Reversing the process, adding edges to a two-colorable graph means that in general, we act with a non-local CZ operation between vertices that have the same color, and thus we create a three-colorable graph. This means that we are acting with a non-local operator on a state that is initially described by an OA. This implies that every three-colorable graph is described as a QOA since the so-called ``quantumness of the QOA" is modified as is mentioned in \cite{zahra_qoa}. We refer the reader to the aforementioned work for a detailed analysis of this concept.

Another remark regarding proposition \ref{thm-3col-general} is about the Schmidt measure of any three-colorable graph. As we have already explained in section \ref{section:2-col}, this discussion is essentially the same as finding the minimal number of terms required to describe the state in the Z-eigenbasis. Connecting with the fact that every three-colorable graph is a QOA, finding the minimum of the Schmidt measure means we have to determine the minimum number of rows required to describe the chosen three-colorable graph state as a QOA. 
Initially, in section \ref{section:2-col} it was established that
\begin{equation}
\label{eq:SM3col-1}
    \text{min}\left(n_R,n_B\right)\leq E_s(\ket{\psi_{\text{2-color}}}).
\end{equation}
A similar thought process can be followed also in this case if we assume two bipartitions of the system. In the first bipartition we assume that we have the red and green colored particles, while for the second one we assume only the blue particles. Instead of considering without loss of generality that $n_R\le n_G\le n_B$, we can assume a stronger relationship that is $n_R+n_G\le n_B$. If this is the case, then it is straightforward to understand, using equation (\ref{eq:SM3col-1}) that:

\begin{equation}
\label{eq:SM3col-2}
    E_s(\ket{\psi_{\text{3-color}}})\ge\text{min}\left(n_R+n_G,n_B\right)= n_R+n_G.
\end{equation}
which demonstrates that our closed-form expression achieves to obtain states that have the minimal Schmidt measure as long as the stricter condition imposed is satisfied. For this discussion, one should remember that for graph states that are equivalent under local unitary operations, this means their Schmidt measure must have the same value. We examine the case where the condition $n_R+n_G\le n_B$ is not satisfied. In this case, we can always trace one of the colors out and the stricter bound is:
\begin{equation}
    \label{eq:SM-3-col-ng-bound}
    n_G\le E_s(\ket{\psi_{\text{3-color}}})
\end{equation}
The above result can be understood intuitively by a very simple example. Let us consider a graph state with three particles. From \cite{Eisert} it is known there is only one LC class for qubits. We can assume two different states, one that creates a circle with 3 particles. In this case, this is the simplest three-colorable graph state. On the other hand, we can assume the same graph, but we can delete one edge, which is a two-colorable graph. For this graph we know that the Schmidt measure is 1, since it is a two-colorable graph with two blue vertices and one red. This means that the Schmidt measure is also 1 for the graph state with the graph the circle with 3 vertices. In the section \ref{section:x-colorable}, we are going to elaborate on this discussion. The crucial note at this point is that as long as the condition $n_R+n_G\le n_B$ is satisfied, the equation (\ref{eq:SM3col-2}) holds. This means that, by changing the base of the logarithm of the Schmidt measure from $\text{log}_2$ to $\text{log}_d$ and by denoting as $N(\ket{\psi_{\text{3-color}}})$ the number of terms in the Z eigenbasis for the three-colorable graph states reads:
\begin{equation}
    d^{n_R+n_G}\leq N(\ket{\psi_{\text{3-color}}}), \quad \text{where}\quad n_R+n_G\leq n_B.
\end{equation}
This means that as long as $n_R+n_G\leq n_B$, the derived closed-form given by \ref{thm-3col-general} provides the minimum number of terms in the Z-eigenbasis.

\section{A special class of three-colorable Graph States }\label{section:3col-special}
We extend our analysis of three-colorable graph states to a special class tailored for the construction of $k$-uniform and AME states. This approach serves two main purposes. Firstly, these states are of particular interest in quantum information applications. Secondly, the presented methodology is an alternative approach to obtaining the closed-form expression for this special class of three-colorable graph states. Moreover, this framework offers an intuitive perspective on determining the equivalence of two graph states under local operators, which we will discuss in detail in Section \ref{section:LT}.
\begin{itemize}
    \item Start with a two-colorable graph as the one assumed in section \ref{section:2-col} for the graph state $\ket{\psi_{\text{2-color}}}$, therefore, in this case, the $n_R\leq n_B$ without loss of generality once again.
    \item Assume that the blue particles are separated into 2 sets named $B_c$ and $B_u$. This means that $B_c\cup B_u = B$ and $B_c\cap B_u = \emptyset$. The number of blue particles in the set $B_c$ is $n_{B_c}$ and the number of blue particles in the set $B_u$ is $n_{B_u}$.
    \item It is crucial to assume that $n_R\leq n_{B_u}$. In principle, we can also impose conditions on the relationship between $n_R$ and $n_{B_c}$, but, as we are going to establish later on, this is not important. Therefore, without loss of generality, we are going to assume that $n_R\leq n_{B_c}$.
    \item Among the particles belonging in the set $B_c$, without omitting any of the connections established in the initial graph, create a new two-colorable graph. This means that now the set of particles, belonging to $B_c$, either remain blue particles or change to become green particles. The set of these green particles is denoted as $G$ and the total number of green particles is denoted as $n_G$. The set of these blue particles is denoted as $B_c\setminus G$ and the total number of these blue particles is denoted as $n_{B_c\setminus G}$. We impose without loss of generality that $n_G\leq n_{B_c\setminus G}$. 
    \item The blue particles belonging to $B_u$ remain unaffected as they were in the 2-colorable graph state $\ket{\psi_{\text{2-color}}}$.
\end{itemize}
The above procedure is summarized in figure \ref{fig:3col-setup}. Although it may look like many restrictions have appeared in the construction of this type of three colorable states, these setups have appeared in various contexts, for example for $k$-uniform states \cite{zahra_adam_kuni} or for cluster states, where a few particles of the same color have been connected like in \cite{guhne_marginals}.
\begin{figure*}[htbp]
    \centering
    \includegraphics[width=1.0\linewidth]{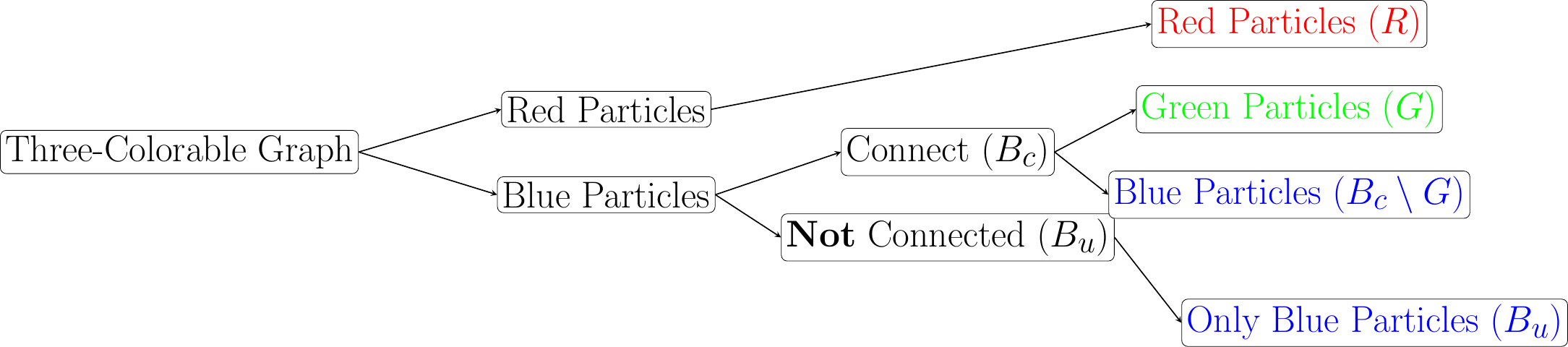}
    \caption{A schematic description of the creation of the three colorable graph states studied in this work.}
    \label{fig:3col-setup}
\end{figure*}
Bearing in mind this information, $\ket{\psi_{\text{3-color}}}$ is defined as follows:
\begin{align}
\label{eq:3coldef}
    \ket{\psi_{\text{3-color}}} = \left(\prod_{b, b' \in B_c} C_b Z_{b'}^{\Gamma_{bb'}}\right)\left(\prod_{b\in B_c,r\in R} C_b Z_{r}^{\Gamma_{br}}\right)\nonumber\\
    \left(\prod_{b\in B_u,r\in R} C_b Z_{r}^{\Gamma_{br}}\right)\ket{+}^{\otimes R}\ket{+}^{\otimes B_u}\ket{+}^{\otimes B_c}.
\end{align}
At this point, understanding the form of the adjacency matrix is crucial to proceed with our main results. Its form is given by equation (\ref{eq:gamma-3col-studied}):
\begin{equation}
\label{eq:gamma-3col-studied}
\Gamma_{3\text{-color}} = \left(\begin{array}{c|c|c|c}
  \smash{\overbrace{\textbf{0}}^{\text{R}}} & \mathbf{\smash{\overbrace{\mathbf{A_{B_uR}}}^{B_u}}} & \mathbf{\smash{\overbrace{A_{B_c\setminus GR}}^{B_c\setminus G}}} & \mathbf{\smash{\overbrace{A_{GR}}^{\text{G}}}} \\ \hline 
  \mathbf{A_{B_uR}^T} & \textbf{0} & 0 & 0 \\ \hline
  A_{B_c\setminus G R}^T & 0 & \textbf{0} & \mathbf{A_{GB_c\setminus G}} \\ \hline
  A_{GR}^T & 0 & \mathbf{A_{GB_c\setminus G}^T} & \textbf{0} \\ 
\end{array}\right).
\end{equation}
Additionally, the comparison of the matrix (\ref{eq:gamma-3col-studied}) with the general form of the adjacency matrix of any three-colorable graph given in equation (\ref{eq:gamma-3col-general}) is unavoidable.

In both cases, the matrices given in equations (\ref{eq:gamma-3col-studied}) and (\ref{eq:gamma-3col-general}) are block matrices. The non-zero blocks denoted with $A$ and a subscript with two or more letters indicating the sets of different colored particles, are assumed they be connected. It is also obvious, that the matrix (\ref{eq:gamma-3col-general}) is a generalization of the matrix (\ref{eq:gamma-3col-studied}). Finally, the reason we denoted with bold letters the part of the matrix (\ref{eq:gamma-3col-studied}) is to make apparent the condition we have imposed, firstly, one has to assume a two-colorable graph, this is the top block with bold letters. Then, if additional connections are assumed, we have the contribution of the additional connections on the two-colorable graph to make a three-colorable graph. Finally, the parts that are non-zero but still not indicated in bold represent the connection of the green and the blue particles with the red ones. If this connection is not maintained, then we would have ended up with two disconnected diagrams, which is not the case for us.

A closing note is that the total number of red particles is denoted by $n_R$, the total number of blue particles is represented with $n_B$, the total number of green particles is expressed as $n_G$, the total number of blue particles belonging to the set $B_u$ is symbolized as $n_{B_u}$ and the total number of blue particles belonging to the set $B_c\setminus G$ is indicated as $n_{B_c\setminus G}$. Finally, yet importantly, it is assumed without loss of generality that $n_R\leq n_B + n_G $ and $n_G\leq n_{B_c\setminus G}$. At this point, the following proposition can be presented:
\begin{proposition}\label{thm-3col}
Let us assume a three-colorable graph state described by the adjacency matrix (\ref{eq:gamma-3col-studied}), defined as in equation (\ref{eq:3coldef}), satisfying the set conditions described in the step-by-step construction at the beginning of this section and without loss of generality $n_R\leq n_{B_u}$, and $n_G\leq n_{B_c\setminus G}$. Then, any such graph state satisfies the following equation:
\begin{align}
\label{eq:3colfinalresultv2special}
     H^{\dagger\otimes B_c\setminus G} H^{\dagger\otimes B_u} \ket{\psi_{\text{3-color}}} = \sum_{\vec{u}=0}^{d-1} \ket{\vec{u}\mathcal{G}_{B_uR}}\Delta\sum_{\vec{g}=0}^{d-1}\ket{\vec{g}\mathcal{G}_{B_c\setminus G}},
\end{align}
where $\vec{u}=(u_1,u_2...,u_{n_R})$, $\vec{g}=(g_1,\cdots,g_{n_G})$ both row vectors, $\mathcal{G}_{B_uR}=\begin{bmatrix} \mathbb{I}_{n_R} & | & A_{B_uR}^T\end{bmatrix}$, $\mathcal{G}_{B_c\setminus G}=\begin{bmatrix} \mathbb{I}_{n_R} & | & A_{GB_c\setminus G}^T\end{bmatrix}$, $f_{k\in B_c} = \sum_{r\in R}\Gamma_{rk}u_r$, and the operator $\Delta$ is defined as following:
\begin{equation}
\label{eq:thm-delta-def}
    \Delta = \left(\bigotimes_{i=1}^{n_G}Z_{g_i}^{f_{g_i}}\right)\left(\bigotimes_{j=1}^{n_{B_c\setminus G}}X_{b_j}^{f_{b_j}}\right).
\end{equation}
\end{proposition}
The proof of proposition \ref{thm-3col} is given in the Appendix \ref{app:proof-3col-thm}. At this juncture, an implication of this proposition must be underlined. Firstly, proposition \ref{thm-3col} can be extended to not only the assumed set-up of the three colorable graph states which is depicted in figure \ref{fig:3col-setup}, where we have only one set of $B_c$ particles. It is possible to have $s$ different numbers of $B_c$-like setups that need to obey the same rules. To make it evident, this idea is depicted in figure \ref{fig:multi-bc}. The extension of the proposition to an arbitrary number of $s$ different $B_c$-like sets is trivial:
\begin{align}
\label{eq:3colfinalresult-multi-bc}
    &  \bigotimes_{k=1}^{s}H^{\dagger\otimes B_{c,k}\setminus G_k} H^{\dagger\otimes B_u} \ket{\psi_{\text{3-color}}} = \notag \\
    & \sum_{\vec{u}=0}^{d-1} \ket{\vec{u}\mathcal{G}_{B_uR}}
    \bigotimes_{k=1}^{s}
    \Delta_{k}\sum_{\vec{g}^k=0}^{d-1}\ket{\vec{g_k}\mathcal{G}_{B_{c,k}\setminus G_k}},
\end{align}
where each term that got a $k$ subscript in the equation (\ref{eq:3colfinalresult-multi-bc}) is trivially extended for any subset $B_{c,k}$ where $k = 1,\cdots,s$ taking into account the definitions given in the equations (\ref{eq:3colfinalresultv2special}) and (\ref{eq:thm-delta-def}).

\begin{figure*}[htbp]
    \centering
    \includegraphics[width=0.7\linewidth]{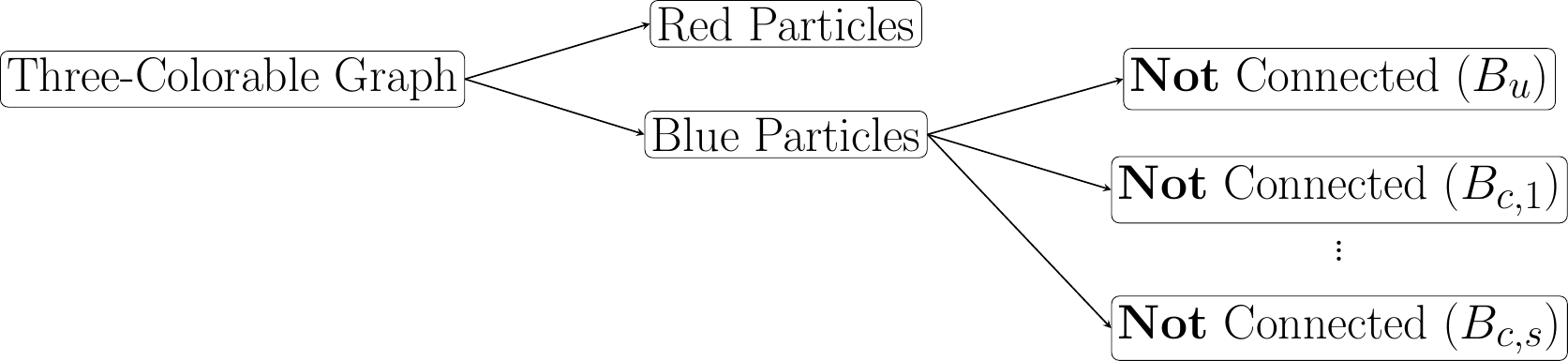}
    \caption{A schematic description of the creation of the most general three-colorable graph state studied in this work. In this case we have multiple $B_c$-like setups ranging from $B_{c,1}$ till $B_{c,s}$}
    \label{fig:multi-bc}
\end{figure*}
At this point, the significance of the proposition \ref{thm-3col} will be discussed. To begin with, similarly to the case of two-colorable graph states we are equipped with a step-by-step methodology to write the closed-form expression of the discussed three-colorable graph states. Let us present the aforementioned guide of how to use the formula:
\begin{itemize}
    \item Step 1: Identify the set every node belongs to, namely one should consider 4 sets: $G,R,B_{u},B_c\setminus G$ and ensure that the conditions $n_R\leq n_{B_{u}}$ and $n_G\leq n_{B_c\setminus G}$ are satisfied.
    \item Step 2: Assign free indices on the particles belonging to the $R$ and $G$ sets. Using these free indices, find the indices for the blue particles belonging in the $B_u$ and the $B_{c}\setminus G$, by adding the indices of their neighbors, precisely like it was done for the 2 colorable cases.
    \item Step 3: To find the $\Delta$ operator one has to apply the $Z$ operations to the green particles of the $B_c$ graph and the $X$ operations to the blue particles. Then the definition of the function $f_{k\in B_c} = \sum_{r\in R}\Gamma_{rk}u_r$ must be recalled.
    \item Step 4: The function $f_{k\in B_c} = \sum_{r\in R}\Gamma_{rk}u_r$ represents the coupling of the red particles with the green and blue particles belonging to the $B_c$ graph. For this reason, to find the functions $f_k$ for each particle in $B_c$, one has to find every red particle that are connected and for each case multiply the index assigned for a red particle in step 2 and multiply it with the weight between these two vertices. Then, for each particle belonging to the $B_c$ set one has to add these coefficients to determine each one of the $f_{k\in B_c}$.
    \item Step 5: Thanks to steps 3 and 4, the $\Delta$ operator is obtained. Therefore, the final step is to get the result from step 2 for the particles belonging in the $B_{u}$ and $R$ sets. Then, place the $\Delta$ operator found in step 4 and finally write the second result of step 2 which is the two-colorable closed-form expression for the $B_c$ subgraph and the formula of the desired state is found. 
\end{itemize}
A general idea of how proposition \ref{thm-3col} works is depicted in figure \ref{fig:3col-des}. 
\begin{figure*}[htbp]
    \centering
    \includegraphics[width=0.8\linewidth]{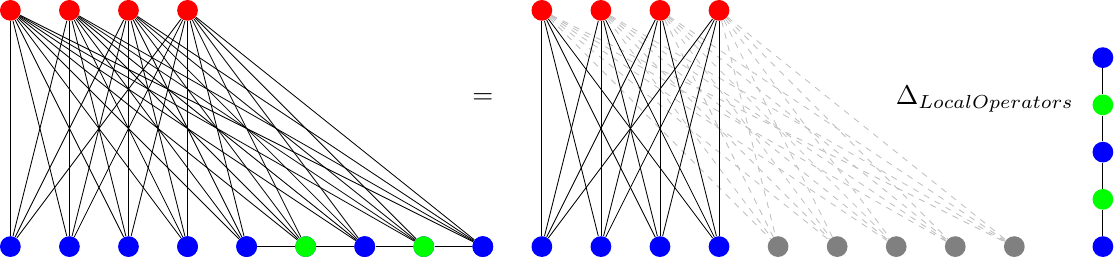}
    \caption{An example schematic representation of closed-form expression given in proposition \ref{thm-3col}. On the left-hand side, the three-colorable graph is shown. On the right-hand side, the two colorable subgraphs are presented. Between them $\Delta_{LocalOperators}$ is applied. This figure illustrates the process of obtaining the family of states satisfying the required conditions to apply the closed-form expression. In this example on the right-hand side, we have a completely bipartite graph (left of the $\Delta$) and a 1D cluster state (right of $\Delta$). However, this idea extends to any state that satisfies the corresponding conditions.}
    \label{fig:3col-des}
\end{figure*}
To make the implementation of the above steps transparent, let us present sequentially the application of proposition \ref{thm-3col} for two representative examples.

\textbf{Example 1: }As the first example, assume the graph state corresponding to the graph given in figure \ref{fig:3col-ex1}. The first step is to ensure the conditions $n_R\leq n_{B_u}$ and $n_G\leq n_{B_c\setminus G}$, which are obeyed since $n_R = 4$, $n_{B_u}=6$, $n_{B_c\setminus G} = 2$ and $n_G = 1$. Step 2 requires assigning free indices to the green and the red particles. Step 3 demands to find the indices for the blue particles, which is simply the addition of the free indices found in step 2 of the neighboring indices. This means, that for the connections involving the particles belonging to the $R$ and the $B_u$ set we have:
\begin{align}
    \sum_{i,j,k,l = 0}^{d-1} &\; \ket{i} \ket{j} \ket{k} \ket{l} \nonumber \\
    &\ket{i+j+k+l} \ket{i+j+k+l} \ket{i+j+k+l}\nonumber \\
    & \ket{i+j+k+l}\ket{i+j+k+l} \ket{i+j+k+l}.
\end{align}
Similarly, we have the term for the particles belonging to the $G$ and the $B_c\setminus G$ set:
\begin{equation}
    \sum_{g=0}^{d-1}\ket{g}\ket{g}\ket{g}.
\end{equation}
Step 3 requires to construct the $\Delta$ operator. For this, we know that on the blue particles of the $B_c$ graph, $X$ gates must be applied, while for the green particle a $Z$ gate. Step 4 requires determining the function $f_{k\in B_c}$, which here is easy since we have assumed that we have weights equal to 1:
\begin{equation}
    f_{k\in B_c} = 1\times i+1\times j+1\times k+1\times l,
\end{equation} 
where the multiplication with one is used to make explicit how the corresponding weight must be used. Bearing the above procedure in mind and combining every step the final state can be written as:
\begin{widetext}
    \begin{align}  
    \label{eq:3col-ex1-final}
        \ket{\psi_{\text{example-}1}} = 
        &\sum_{i,j,k,l = 0}^{d-1} \ket{i} \ket{j} \ket{k} \ket{l} 
        \ket{i+j+k+l} \ket{i+j+k+l} \ket{i+j+k+l} \ket{i+j+k+l} 
        \ket{i+j+k+l} \ket{i+j+k+l} \nonumber \\
        &\left(X_{x_1}^{i+j+k+l} \otimes Z_{x_2}^{i+j+k+l} 
        \otimes X_{x_3}^{i+j+k+l}\right) 
        \sum_{g=0}^{d-1} \ket{g}_{x_1} \ket{g}_{x_2} \ket{g}_{x_3}.
    \end{align}
\end{widetext}
It has to be underlined, that in equation (\ref{eq:3col-ex1-final}), the indices $x_1$ and $x_3$ represent the blue particles of the $B_c$ graph and $x_2$ the green particle. In fact, this example has a dual goal. The first was to demonstrate the detailed way to apply the closed-form formula presented in proposition \ref{thm-3col}. The second one is that if the $B_c$ graph is a one-dimensional cluster state, the closed-form expression becomes rather interesting. As we explain in the Appendix \ref{appen:1d-cluster}, in those cases, if $B_c$ has 2 particles, then the Bell basis is obtained. When we have 3 particles the GHZ basis is obtained. 
\begin{figure*}[htbp]
    \centering
    \begin{minipage}{0.25\linewidth}
        \centering
      \textbf{Step 1} \\[5pt]
        \includegraphics[width=0.5\linewidth]{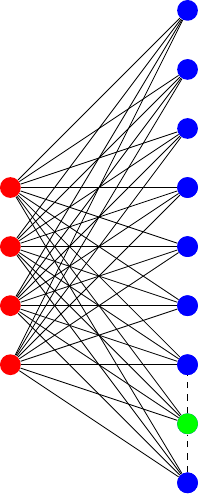} 
      
    \end{minipage}\hfill
    \begin{minipage}{0.25\linewidth}
        \centering
        \textbf{Step 2} \\[5pt]
        \includegraphics[width=0.7\linewidth]{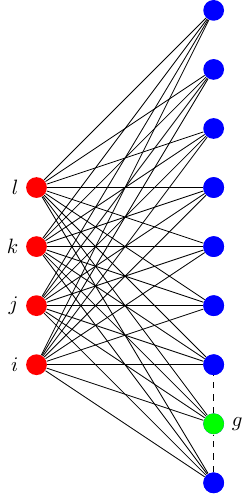} 
    \end{minipage}\hfill
    \begin{minipage}{0.25\linewidth}
        \centering
        \textbf{Step 3} \\[5pt]
        \includegraphics[width=\linewidth]{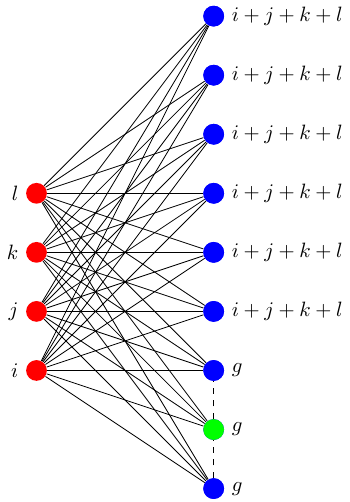} 
    \end{minipage}
    \caption{Example 1: Presentation of the first three steps of the algorithm to obtain the closed-form expression (\ref{eq:3col-ex1-final}) using proposition \ref{thm-3col}.}
    \label{fig:3col-ex1}
\end{figure*}

\textbf{Example 2:} Let us proceed with the presentation of the second example presented in figure \ref{fig:3col-ex2}. Here in this example, we have two different $B_c$ sets. This does not have any impact on the application of the step-by-step procedure of the algorithm. In this case, the condition $n_R\leq n_{B_u}$, $n_{G,1}\leq n_{B_c\setminus G,1}$, and $n_{G,2}\leq n_{B_c\setminus G,2}$, which they do since $n_R = 4$, $n_{B_u}=5$, $n_{B_c\setminus G,1} = 1$, $n_{G,1} = 1$, $n_{B_c\setminus G,2} = 1$ and $n_{G,2} = 1$. The second step is to assign the free indices to the red and the green particles. Here, we did so for the red particles, but we have assigned the free indices to the blue particles to make it crystal clear that if $n_{G,s} = n_{B_c\setminus G,s}$ for whatever $s$, then it is irrelevant where the indices will be assigned. In any case, one has to apply them to the same group of colors. Step 3 is straightforward and presented in figure \ref{fig:3col-ex2}. So, we have obtained the summations we intended to. For the connection between the $R$ and $B_u$ particles we have:
\begin{align}
    \sum_{i,j,k,l = 0}^{d-1} &\; \ket{i} \ket{j} \ket{k} \ket{l} \nonumber \\
    &\ket{i+j+k+l} \ket{i+j+k+l} \ket{i+j+k+l}\nonumber \\
    & \ket{i+j+k+l}\ket{i+j+k+l} ,
\end{align}
while for the $B_{c,1}$ and $B_{c,2}$ we have:
\begin{equation}
    \sum_{b_s=0}^{d-1}\ket{b_s}\ket{b_s},\quad \text{where} \quad s=1,2.
\end{equation}
Step 3 requires to construct the $\Delta$ operator. For this, we know that on the blue particles of each $B_c$ graph, $X$ gate must be applied, while for the green particles $Z$ gate. Step 4 requires determining the function $f_{k\in B_c}$, for both $B_c$ sets assumed in this, which here is easy since we have assumed that we have weights equal to 1:
\begin{equation}
    f_{k\in B_{c,s}} = i+j+k+l,\quad \text{where} \quad s=1,2.
\end{equation} 
Therefore, combining the above results for the graph states depicted in figure \ref{fig:3col-ex2}, the final form is obtained:
\begin{widetext} 
    \begin{align}
    \label{eq:3col-ex2-final}
        \ket{\psi_{\text{example-}2}} = 
        &\sum_{i,j,k,l = 0}^{d-1} \ket{i} \ket{j} \ket{k} \ket{l} 
        \ket{i+j+k+l} \ket{i+j+k+l} \ket{i+j+k+l} \ket{i+j+k+l} 
        \ket{i+j+k+l} \nonumber \\
        &\left(X^{i+j+k+l}\otimes Z^{i+j+k+l} 
         \sum_{b_1=0}^{d-1} \ket{b_1} \ket{b_1}\right)\otimes\left(X^{i+j+k+l}\otimes Z^{i+j+k+l} 
         \sum_{b_2=0}^{d-1} \ket{b_2} \ket{b_2}\right).
    \end{align}
\end{widetext}

\begin{figure*}[htbp]
    \centering
    \begin{minipage}{0.25\linewidth}
        \centering
      \textbf{Step 1} \\[5pt]
        \includegraphics[width=0.5\linewidth]{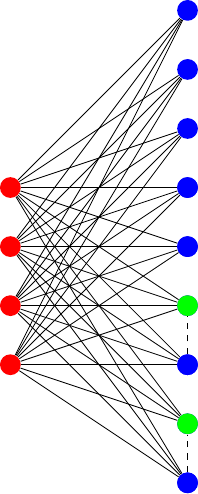}    
      
    \end{minipage}\hfill
    \begin{minipage}{0.25\linewidth}
        \centering
        \textbf{Step 2} \\[5pt]
        \includegraphics[width=0.7\linewidth]{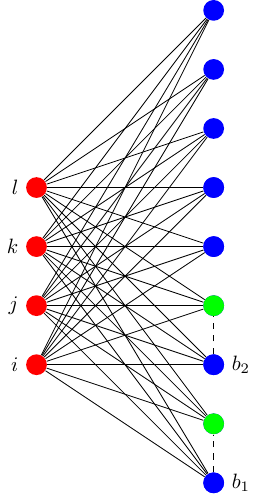} 
    \end{minipage}\hfill
    \begin{minipage}{0.25\linewidth}
        \centering
        \textbf{Step 3} \\[5pt]
        \includegraphics[width=\linewidth]{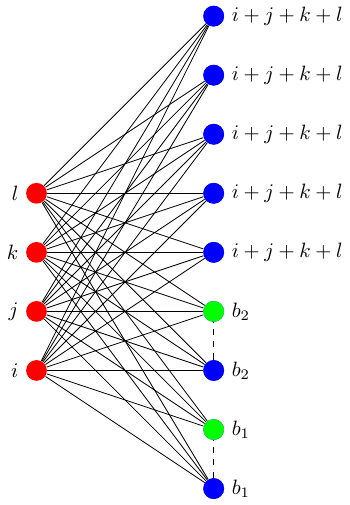} 
    \end{minipage}
    \caption{Example 2: Presentation of the first three steps of the algorithm to obtain the three-colorable closed-form expression (\ref{eq:3col-ex2-final}) using proposition \ref{thm-3col}.}
    \label{fig:3col-ex2}
\end{figure*}

Following the thought process of the implications of the two and three-colorable results, it is unavoidable to discuss the Schmidt measure, for the theorized graph state setup case. As it is explained in the corresponding part of sections \ref{section:2-col} and \ref{section:3col-special} this discussion holds for any local dimension $d$. A bound on Schmidt measure is already discussed in the corresponding part of \ref{section:3col-gen} and it is true also in this case as long as the stricter condition explained is imposed. For this, one should take into consideration the discussion after proposition 7 in \cite{Eisert}. The reasoning is that every graph, which is not two-colorable can become two-colorable if an adequate number of odd circles are deleted since this corresponds to $Z$ measurements. Therefore, we can adapt the equation (\ref{eq:SMbound1}) as follows:
\begin{equation}
\label{eq:SMbound1-3col}
     E_s(\ket{\psi_{\text{3-color}}})\leq\lfloor\frac{n+K}{2}\rfloor,
\end{equation}
where $K$ is the number of vertices deleted $n_{B_u}+n_R+n_G+n_{B_c\setminus G} = n$.

It straightforward to grasp that for a graph state with adjacency matrix $\Gamma_{\text{3-color}}$ given by equation (\ref{eq:gamma-3col-studied}) the lower bound on Schmidt measure, should be at least dictated by the $\text{min}\left(n_R,n_{B_u}\right)$ and $\text{min}\left(n_{B_c\setminus G},n_G\right)$ therefore we have that:
 \begin{equation}
\text{min}\left(n_R,n_{B_u}\right)+\text{min}\left(n_{B_c\setminus G},n_G\right)\leq E_s(\ket{\psi_{\text{3-color}}}),
\end{equation}
which by our construction leads to:
\begin{equation}
\label{eq:SMbound3-3col}
n_R+n_G\leq E_s(\ket{\psi_{\text{3-color}}}).
\end{equation}
The above equation holds in the assumed setup because we have imposed that $n_R\le n_{B_u}$ and $n_G\le n_{B_c\setminus G}$ which leads to the condition $n_R+n_G\le n_B$. Therefore combining equations (\ref{eq:SMbound1-3col}) and (\ref{eq:SMbound3-3col}):
\begin{equation}
\label{eq:SMbound-3col-final}
     n_R+n_G\leq E_s(\ket{\psi_{\text{3-color}}})\leq\lfloor\frac{n + K}{2}\rfloor.
\end{equation}
This result is crucial for the impact of proposition \ref{thm-3col} since the number of the terms we need to express the graph state with adjacency matrix (\ref{eq:gamma-3col-studied}) was found to be $d^{n_R}\times d^{n_G} = d^{n_R+n_G}$. By the above discussion, we explained that for the special class of states studied, the derived closed-form expression has the minimal number of terms in the Z-eigenbasis

The described three-colorable graph state with adjacency matrix, given by equation (\ref{eq:gamma-3col-studied}), is, according to the proven proposition \ref{thm-3col}, on the form of a QOA, therefore every scenario, in which the QOAs are utilized, becomes automatically an interesting area for the assumed family of three-colorable graph states. To be more precise, we refer the reader to the equation (38) of the reference \cite{zahra_qoa}. In this work, it is explained that we can multiply the so-called quantum columns of a QOA. This idea is also apparent in the equation \ref{eq:3colfinalresult-multi-bc} of this work. In section \ref{section:3col-gen}, we have explained that every three-colorable graph state is a QOA. In this section, by bearing in mind the equation (\ref{eq:3colfinalresult-multi-bc}) we conclude the same result. In fact, with this methodology, we are presenting a general way to construct QOAs, since we just have to assume several $B_c$-like graphs that either have 2 or 3 particles. In reality, we suspect that we can construct states that are QOAs by concatenating Bell, GHZ, and other AME states of a lower number of particles in the quantum part of the QOA, but this is a matter of future research. Graphically this means one can use multiple $B_c$-like graphs in the original two-colorable one. The second reason why this setup is important is connected with the construction of $k$-uniform states. For this, not only the discussion in this paragraph must be considered, but also the results presented in \cite{zahra_adam_kuni}, but these will be presented in \ref{section:k-uni}. 

A final remark is about the initial condition imposed on the number of $n_R$ and $n_{B_{u}}$ particles, which must satisfy $n_R\leq n_{B_u}$. For this, we have to recall the fact that we can always make an arbitrary graph, two-colorable if the adequate number of odd-number circles are removed from a non-two-colorable graph. In this case, one has to recall the equation (\ref{eq:SMbound1-3col}) and the fact that $K$ is the number of vertices removed. What we carried out in practice was the reverse process, where we started with a two-colorable graph and we started adding new edges. The physical implication of the demand that the bound $n_R\leq n_{B_u}$ is not violated, is that in case it is, then the entanglement properties of the new three-colorable graph state have changed radically and the correlation with the entanglement properties of the initial two-colorable case is no longer apparent. This result is completely consistent with the fact that the entanglement properties of a state can change only with joint operators. Last but not least, recalling that almost all two-colorable graph states have maximal Schmidt measure \cite{Severini_2006} for the qubit case, it is rational to examine if it is possible to transform the three-colorable graph state we examined to a two-colorable state. This question will be addressed in the next section.

\section{Local Transformations between Two- and Three-Colorable Graph States}\label{section:LT}

\subsection{Short overview on local transformations}
Entanglement characterization is crucial in the study of graph states, making the examination of equivalence classes a natural problem. Moreover, graph states frequently arise in quantum networks, where each vertex represents a spatially separated qudit or laboratory, preventing joint operations \cite{guhne_marginals}.  Consequently, quantum operations $\mathcal{E}$ must be considered as a subclass of completely positive maps (CPMs) that remain separable under the finest partitioning. We are interested in the transformation from state $\ket{\psi_1}$ to $\ket{\psi_2}$ with a non-zero probability, if a CPM $\mathcal{E}$ is adopted. Characterizing the complete class of every transformation $\mathcal{E}$ which belongs to the Local Operations and Classical Communications (LOCC) is broad and generally challenging. Usually, it is assumed that:
\begin{equation}
    \mathcal{E}(\rho) = \bigotimes_{j=1}^{N}\mathcal{E}_j\rho\bigotimes_{j=1}^{N}\mathcal{E}_j^{\dagger},
\end{equation}
which means that: $\ket{\psi_2} = \bigotimes_{j=1}^{N}\mathcal{E}_j\ket{\psi_1}$. For clarification:
\begin{equation}
    \bigotimes_{j=1}^{N}\mathcal{E}_j = \mathcal{E}_1\otimes\cdots\otimes\mathcal{E}_N.
\end{equation}
At this point, let us present three dissimilar classes of local operations for the case of local dimension $d=2$ \cite{gstates_review}. The first one is the invertible Stochastic Local Operations assisted by SLOCC \cite{GHZ_vs_W}, which means that the operation performed in each qudit is of the form $\mathcal{E}_i\in\text{SL}(2,\mathbb{C})$. The probability of achieving an SLOCC transformation is typically less than one. Then, the LU equivalence must be taken into account, which means that each operation on every qudit is $\mathcal{E}_i\in U(2)$. Finally, the LC are operations on each qudit such that $\mathcal{E}_i\in \mathcal{C}_1$, which means they are one of the Clifford unitaries \cite{quditsLCbahramgiri2007graphstatesactionlocal}. The last two have the probability to be achieved, if they exist, equal to unity.
At this point, deepening our focus on the LC equivalence is vital to understanding the consequences of propositions \ref{thm-2col} and \ref{thm-3col}. The generalized Pauli group is generated as follows \cite{quditsLCbahramgiri2007graphstatesactionlocal}:
\begin{equation}
    \mathcal{P} = \{\omega^aX^bZ^c\},\quad \text{with} \quad a,b,c\in\mathbb{F}_d.
\end{equation}
The n-body Pauli group $\mathcal{P}_n$ is defined as the tensor product of $\mathcal{P}$. Then, the Clifford group for n particles $\mathcal{C}_n$ can be presented, which is the normalizer of the $\mathcal{P}_n$.
\subsection{Local Complementation, LC equivalence, and two-colorable closed-form expression}
It is noteworthy that for qubits, there is a simple graphical rule to determine and examine if two graph states are LC equivalent \cite{Van_den_Nest_2004} and is called Local Complementation \cite{BOUCHET199375}. The rule is rather simple, one has to choose a vertex, let us assume we call it $A$. Then, the neighborhood $N_A$ of $A$ must be checked. If two vertices are connected, the edges are deleted, if they are not, a connection is established. Additionally, it was shown in \cite{Eisert} that the aforementioned rule is as follows:
\begin{equation}
    \label{eq:lcqubits}
    \text{LC}_A = \sqrt{-iX_A}\bigotimes_{b\in N_A}\sqrt{iZ_b}.
\end{equation}
The idea of a graphical rule to determine LC equivalency has been extended to qudits \cite{quditsLCbahramgiri2007graphstatesactionlocal}, and the corresponding local transformations are found in \cite{qudit_LC_1,qudit_LC_2}. It must be noted that, if two graph states are not LC equivalent, this does not imply that they are also not LU equivalent if the number of particles is above 8, as it was shown in \cite{ji2008lulcconjecturefalse}. In \cite{guhne_lclu}, an approach to construct examples for this case is presented. Finally, one has to bear in mind that the application of local Clifford operations on a graph state does not necessarily lead to a graph state directly, but in general, it should be a stabilizer state \cite{guhne_marginals}. But since stabilizer states can be realized as graph states, as it was shown in \cite{eisert_schlingemann2001}, every stabilizer state is LC equivalent to some graph state. A final remark is that, even if an efficient algorithm for graphs has been developed \cite{BOUCHET199375}, as well as the algorithm for the determination of two stabilizer states are LC equivalent, the endeavor of increasing the number of particles and finding the LC orbit for a large number of particles turns out to be a challenging goal \cite{lcrobits}.

In any case, and especially for a small number of particles, it is always a good idea to find whether a three-colorable graph can be transformed into a two-colorable one by local complementation and then apply the formula found in proposition \ref{thm-2col}. Additionally, in scenarios where the closed-form formula for the three-colorable states can not be used, for instance, if $n_{B_u}\leq n_R$, local complementation can be used to obtain a two-colorable graph and then apply our methodology. An example, in this case, is given by figure \ref{fig:LC-ex}, where local complementation on the particle indicates that $k$ eliminates the connection between the particles $b_1$ and $g$, transforming it into a two-colorable graph state, where the closed-form formula is known. It should be noted that the LC equivalency is something we can suspect from the fact that the Schmidt measure, for the left and right states, is the same. In the case of two colorable graphs, the red particles are 3 so the Schmidt measure is $3$ and the minimal number of terms in the Z-eigenbasis is $d^3$. In the case of the three colorable graph states we have to modify, without loss of generality, the hierarchy for the number of vertices with respect to the color as $n_G\le n_B\le n_R$. In this case, the Schmidt measure is $n_G+ n_B$, which is again $3$ and the minimal number of terms in the Z-eigenbasis is $d^3$.
\begin{figure}
    \centering
    \includegraphics[width=0.8\linewidth]{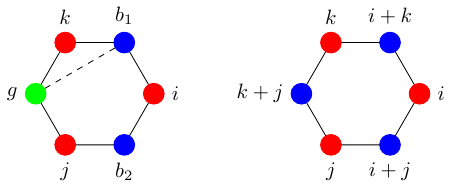}
    \caption{An example usage of the Local Complementation to transform a three-colorable graph to a two-colorable graph. On the left, the initial three-colorable graph is given. On the right, the implementation of the method for the closed-form expression is performed.}
    \label{fig:LC-ex}
\end{figure}

\subsection{Association of the presented framework with $k$-uniform states}\label{section:k-uni}
The goal of this subsection is to indicate that the three-colorable graph states, assumed in \ref{section:3col-special}, with the adjacency matrix given by equation \ref{eq:gamma-3col-studied}, arise naturally in the case of the $k$-uniform states. In a few words, $k$-uniform states are highly entangled pure states and their property is that all of their k-qudit reduced density matrices are maximally mixed \cite{kuni-1,kuni-2,kuni-happy-paper,zahra_adam_kuni,QOA-def1,zahra_qoa}. The way to construct the $k$-uniform graph states, presented in \cite{zahra_adam_kuni}, fits exactly into the framework of the states presented in \ref{section:3col-special}. To elaborate on that, it is stated that one has to assume graph states with adjacency matrix given by the following matrix:
\begin{equation}
\label{eq:gamma-k-uni}
\Gamma_{k\text{-uniform}} = \left(\begin{array}{c|c} 0 & A \\ \hline A^T & B \end{array}\right).
\end{equation}
The only requirement is that the matrix $A$, as well as each one of its submatrices, are non-singular while the block matrix $B$ is not limited by anything. It is clear that for $B=0$, a two-colorable graph state is obtained and in proposition \ref{thm-2col}, the closed-form expression has been given. Additionally, if it is assumed that the matrix $B$ has the following structure:
\begin{equation}
\label{eq:gamma-B}
B = \left(\begin{array}{c|c} 0 & B_1 \\ \hline B_1^T & 0   \end{array}\right),
\end{equation}
where $B_1$ and each one of its sub-matrices is non-singular, the so-called hierarchical graph states are created. In this case, the graph state assumed, is simply a special case of the three-colorable graph state described by the adjacency matrix (\ref{eq:gamma-3col-studied}). In \cite{zahra_adam_kuni}, two vital results regarding the SLOCC equivalence are described. In their work, it is interpreted under the scope of $k$-uniform states, but as we explained, the colorability understanding is essentially the same. In both cases, the conditions regarding the $A$ block of the matrix (\ref{eq:gamma-k-uni}) must hold. 

They showed that the $k$-uniform state with adjacency matrix given by (\ref{eq:gamma-k-uni}) and $B=0$, and the state, given by the same adjacency matrix, and the block $B$ given by matrix (\ref{eq:gamma-B}), do not belong in the same SLOCC class. In our perspective, this means that we need to consider a general two-colorable graph state described by the matrix (\ref{eq:gamma-2-col}) and to assume that every submatrix of $A_{RB}$ is nonsingular. The second stage is to assume we have the three colorable graphs using our framework described by the matrix (\ref{eq:gamma-3col-studied}) and just impose that every submatrix of $G_{B_c\setminus G}$ is non-singular. Then, these two states do not belong to the same SLOCC class. 

The second example is to assume once again a two-colorable graph state, described by the adjacency matrix (\ref{eq:gamma-2-col}), and impose the conditions to have a $k$-uniform state. Then, the second graph state is the same as the first one, but the $B_c$ graph has only two particles. Then, if the total number of particles is odd, these two states are not SLOCC equivalent. These two examples showcase that, indeed in some cases, it is not possible to transform the theorized 3 colorable graph states described by the adjacency matrix (\ref{eq:gamma-3col-studied}) to their corresponding two-colorable graph state with the adjacency matrix given by (\ref{eq:gamma-2-col}).

\section{$\chi$-Colorable graph states}\label{section:x-colorable}
Generalizing the formalism developed so far, we extend our analysis to graph states with arbitrary colorability $\chi$. This requires considering graphs with a chromatic number $\chi$, where the vertex set is partitioned into $\chi$ disjoint color classes. Instead of assigning specific names to the colors, we denote them as color 1, color 2, up to color $\chi$, forming corresponding sets of vertices labeled as $c_{1}, c_{2}, \dots, c_{\chi}$. Similarly, the number of vertices in each set $c_l$ (where $l \in \{1,2,\dots,\chi\}$) is denoted as $n_l$. Without loss of generality, we impose the ordering $n_1\le n_2\le\cdots\le n_{\chi-1}\le n_{\chi}$. The total number of vertices in the graph is then given by:
\begin{equation}
n = \sum_{j=1}^{\chi} n_j\, .
\end{equation}

The adjacency matrix of a $\chi$-colorable graph state is structured as follows:

\begin{equation}
\label{eq:gamma-chi-col}
\Gamma_{\chi\text{-color}} = \left(\begin{array}{c|c|c|c}
  \smash{\overbrace{0}^{c_1}} & \mathbf{\smash{\overbrace{A_{c_2,c_1}}^{c_2}}} & \mathbf{\smash{\overbrace{\cdots}^{\cdots}}} & \mathbf{\smash{\overbrace{A_{c_{\chi},c_1}}^{c_{\chi}}}} \\ \hline 
  A_{c_2,c_1}^T & 0 & 
  \begin{array}{c}
  \vdots \\ 
  \ddots \\ 
  \cdots
  \end{array} & \vdots \\ \hline
 \vdots & 0 & \ddots & A_{c_{\chi},c_{\chi-1}} \\ \hline
  A_{c_{\chi},c_1}^T & 0 & A_{c_{\chi},c_{\chi-1}}^T & \textbf{0} \\ 
\end{array}\right).
\end{equation}
A $\chi$-colorable graph state is defined as:
\begin{equation}
\label{eq:x-col-state}
    \ket{\psi_{\chi\text{-color}}} = 
    CZ_{\chi} 
    \ket{+}^{\otimes c_1}\ket{+}^{\otimes c_2}\cdots\ket{+}^{\otimes c_{\chi}}\, .
\end{equation}
where $CZ_{\chi}$ is defined as:
\begin{equation}
    CZ_{\chi}= \prod_{1 \leq i < j \leq \chi} \prod_{p\in c_i, q\in c_j} C_{p}Z_{q}^{\Gamma_{p,q}}.
\end{equation}
For our upcoming analysis, it is also useful to define the following quantity,
\begin{equation}
    \label{eq:z-x-gen-x-comp}
    Z_{\chi-1} = \bigotimes_{i=2}^{\chi-1} \bigotimes_{j=1}^{n_{c_i}} Z_{c_i,j}^{\sum_{k=1}^{i-1} \overrightarrow{v}^{c_k} \cdot \mathbf{(A_{c_i,c_k})_{c_i,j}}},
\end{equation}
where the object $\mathbf{(A_{c_i,c_k})_{c_i,j}}$ is the generalization of the vector defined in proposition \ref{thm-3col-general}. Namely, this symbolizes the vector created by the extraction of the column of the block matrix $A_{c_i,c_k}$ that corresponds to the $j$-th particle with color $c_i$. We proceed with our analysis by defining a set of vectors as follows:
\begin{equation}
    \label{eq:x-col:vectors}
    \vec{v}^{c_m} = (v_1^{c_m},\cdots,v_{n_{c_m}}^{c_m}), \quad \text{where}, \quad m\in\{1,\cdots,\chi\}.
\end{equation}
Let us denote the concatenation of every vector $\vec{v}^{c_l}$, with $l\in\{1,\cdots,\chi-1\}$, starting from $\vec{v}^{c_1}$ till $\vec{v}^{c_{\chi-1}}$ as follows:
\begin{equation}
    \label{eq:x-col-vt}
    \vec{V}_{\chi}= (\vec{v}^{c_1},\cdots,\vec{v}^{c_{\chi-1}})
\end{equation}
It will also be useful to note that we are going to denote as $\mathbb{I}_{\chi}$ the identity matrix of dimension $\sum_{j=1}^{\chi-1}n_{c_j}$. Finally, we define a generator-like matrix $\mathcal{G}_{\chi}$, with dimensions $\left(\sum_{j=1}^{\chi-1}n_{c_j}\right)\times \left(\sum_{j=1}^{\chi}n_{c_j}\right)$ as:
\begin{equation}
    \label{eq:Generator-matrix-c-col}
    \mathcal{G}_{\chi} = \begin{bmatrix} 
    \mathbb{I}_{\chi} & \Bigg| & \begin{bmatrix} A_{c_{\chi},c_1} \\ \vdots \\ A_{c_{\chi},c_{\chi-1}} \end{bmatrix}
    \end{bmatrix},
\end{equation}
Bearing the above definitions into mind, we have the following proposition, which is a direct generalization of the results obtained in the previous sections.
\begin{proposition}\label{thm-x-col}
A qudit $\chi$-colorable graph state defined in equation \ref{eq:x-col-state} with the corresponding adjacency matrix of the graph defined in (\ref{eq:gamma-chi-col}) can be written as:
\begin{equation}
    \label{eq:CFE-every-colorability}
    H^{\dagger\otimes c_{\chi}}\ket{\psi_{\chi\text{-color}}} = \sum_{\vec{V}_{\chi}=0}^{d-1}Z_{\chi-1}\ket{\vec{V}_{\chi}\cdot \mathcal{G}_{\chi}},
    \end{equation}
    with every element used in this proposition defined from the beginning of this section.
\end{proposition}
\begin{proof}
    The proposition \ref{thm-x-col} is remarkably straightforward to prove, if the same steps of the proof for proposition \ref{thm-3col-general} are adopted.  
\end{proof}
It is notable that for $\chi=3$, we are reproducing the closed-form expressions presented in the propositions \ref{thm-2col} and \ref{thm-3col-general} respectively. Following the thought process of displaying the implications of each proposition presented, we start by equipping the reader with the step-by-step guide to write the closed-form expression of a $\chi$-colorable graph state.
\begin{itemize}
    \item Step 1: Identify the color of each vertex, with colors starting from $c_1$ till $c_{\chi-1}$ and assign free indices to them.
    \item Step 2: The indices of the vertices with color $c_{\chi}$ will be the summation of the neighboring indices with color $c_l$, where $l\in\{1,\cdots,\chi-1\}$ multiplied with the corresponding element of the adjacency matrix given by equation (\ref{eq:gamma-chi-col}).
    \item Step 3: The exponent of the Z operator for every vertex with color $c_2$ will be the multiplication of their free indices with the free indices of the color $c_1$ weighted with the corresponding element of the adjacency matrix given by equation (\ref{eq:gamma-chi-col}).
    \item Step 4: We proceed iteratively with the exponents of the Z gates for every element with a color $c_l$  where here $l\in\{2,\cdots,\chi-1\}$ and we multiply their free indices with the free indices of every vertex with color $c_{l-1}$  where $l\in\{2,\cdots,\chi-1\}$.
    \item Step 5: We sum over every element of the concatenated vector $\vec{V}_{\chi}$, which has been defined in equation (\ref{eq:x-col-vt}).
\end{itemize}
Applying the same philosophy as to why every three-colorable is a QOA, it is straightforward to understand that, given the definitions of $\vec{V}_{\chi}$ (\ref{eq:x-col:vectors}) and $\mathcal{G}_{\chi}$ (\ref{eq:Generator-matrix-c-col}), the mathematical object $\vec{V}_{\chi}\cdot\mathcal{G}_{\chi}$ forms an OA. Taking into account that we have to apply $Z_{\chi}$ and sum-over, we can understand that a $\chi$-colorable graph state can be written as a QOA with the following dimensions: 
\begin{equation}
\text{QOA}\left(\sum_{j=1}^{\chi-1}n_{c_j},\sum_{j=1}^{\chi}n_{c_j},d,k\right)= \text{QOA}(n-n_{\chi},n,d,k),
\end{equation}
where $k$ satisfies the following relation \cite{zahra_qoa}:
\begin{equation}
    k\, \text{Tr}_{p_1,\cdots,p_{n-k}}\left(\ket{\psi_{\chi-\text{color}}}\bra{\psi_{\chi-\text{color}}}\right)=\left(\sum_{j=1}^{\chi-1}n_{c_j}\right)\mathbb{I}_k,
\end{equation}
for every subset of $n-k$ parties $\{p_1,\cdots,p_{n-k}\}$.

In this work, we have already described multiple times how to derive strict lower bounds on the Schmidt measure of a graph state and how to connect it with the minimum number of terms needed to write the state in the Z-eigenbasis. Therefore it is straightforward to state that,
\begin{equation}
    \label{eq:SM-x-col}
    \sum_{j=1}^{\chi-1}n_{c_j}\le E_s(\ket{\psi_{\chi-\text{color}}}),
\end{equation}
as long as,
\begin{equation}
\label{eq:SM-x-col-con}
    \sum_{j=1}^{\chi-1}n_{c_j}\le n_{c_{\chi}}.
\end{equation}

Equipped with these results, we point out some vital remarks regarding the lower bounds on Schmidt measure. It is well known that the graph states with a star-shape graph ($\chi=2$) is LC equivalent to the GHZ state \cite{gstates_review}. It is also trivial to understand that the completely connected graph with $n$ vertices has $\chi=n$. It is well established that the completely connected graph state is LC equivalent to the GHZ state, which means that it is also LC equivalent to the star-shaped graph. Therefore, these three states must have the same Schmidt measure. One should keep in mind that graph states corresponding to the completely connected graph do not satisfy equation (\ref{eq:SM-x-col}) since the condition (\ref{eq:SM-x-col-con}) is not fulfilled. This makes the discussion very interesting because it illustrates that the endeavor of finding the stricter lower bound for every $\chi-$colorable graph state is a much more complicated task. This endeavor is a matter of future research and we hope that the insights given in this work will assist in this.

\section{Conclusions and outlook}\label{section:conclusion}

In this work, we have established a systematic framework for deriving closed-form expressions for two- and three-colorable graph states, providing an efficient and intuitive method to represent these states with minimal terms in the $Z$-eigenbasis. We further investigated a special subclass of three-colorable graph states, revealing their natural connection to $k$-uniform states. Extending our framework, we formulated a general approach for constructing and analyzing $\chi$-colorable graph states. By leveraging graph state colorability, we demonstrated deep connections between QOAs, LC equivalence classes, and multipartite entanglement structure. Our results show that every two-colorable graph state is LC-equivalent to a state expressible via an OA, while all graph states with $\chi>2$ are LC-equivalent to a QOA, establishing a direct link between combinatorial designs and quantum information theory. 

Beyond formal classification, our findings provide practical tools for constructing highly entangled quantum states. The introduced framework enables a compact representation, with direct applications in quantum networks, quantum error correction, and measurement based quantum computing. Additionally, our generalization to arbitrary $\chi$-colorability opens a pathway to systematically characterizing multipartite entanglement classes in graph-based quantum states.

This work opens multiple research directions. While we derived lower bounds for two- and three-colorable graph states and extended these results to higher $\chi$, an open challenge remains to determine the tightest possible bounds for $\chi$-colorable graph states and their implications for entanglement distribution in quantum networks. Another fundamental question is the refinement of QOAs. Given that we provided evidence that every OA corresponds to a two-colorable graph state, while QOAs correspond to higher chromatic number states, a deeper mathematical exploration is needed to tighten the connection between QOAs and multipartite entanglement.

By integrating concepts from graph theory, combinatorial designs, and quantum entanglement theory, our work lays a conceptual and practical foundation for the systematic study of graph states. The link between colorability, orthogonal arrays, and multipartite entanglement provides a novel perspective for constructing and analyzing highly entangled states. We anticipate that these results will be instrumental in advancing quantum state classification, quantum error correction, and scalable quantum computing architectures.

\acknowledgements

We thank Edwin Barnes, Sophia Economou, Mario Flory, Markus Grassl, Otfried G\"{u}hne, Barbara Kraus, Jan Sperling, G\'{e}za T\'{o}th, and Karol \.{Z}yczkowski for their valuable discussions and remarks. This work was supported by the Equal Opportunity Program, Grant Line 2: Support for Female Junior Professors and Postdocs through Academic Staff Positions, 14th funding round of Paderborn University.

\bibliography{Colorable-Graph}

\begin{thebibliography}{49}
\providecommand{\natexlab}[1]{#1}
\providecommand{\url}[1]{\texttt{#1}}
\expandafter\ifx\csname urlstyle\endcsname\relax
  \providecommand{\doi}[1]{doi: #1}\else
  \providecommand{\doi}{doi: \begingroup \urlstyle{rm}\Url}\fi

\bibitem[Horodecki et~al.(2009)Horodecki, Horodecki, Horodecki, and Horodecki]{Horodecki_2009}
Ryszard Horodecki, Paweł Horodecki, Michał Horodecki, and Karol Horodecki.
\newblock Quantum entanglement.
\newblock \emph{Reviews of Modern Physics}, 81\penalty0 (2):\penalty0 865–942, June 2009.
\newblock ISSN 1539-0756.
\newblock \doi{10.1103/revmodphys.81.865}.
\newblock URL \url{http://dx.doi.org/10.1103/RevModPhys.81.865}.

\bibitem[Datta et~al.(2022)Datta, Kondra, Miller, and Streltsov]{intro-QEC-1}
C.~Datta, Tulja~Varun Kondra, Marek Miller, and A.~Streltsov.
\newblock Entanglement catalysis for quantum states and noisy channels.
\newblock \emph{Quantum}, 2022.

\bibitem[Scott(2004)]{Scott-2004}
A.~J. Scott.
\newblock Multipartite entanglement, quantum-error-correcting codes, and entangling power of quantum evolutions.
\newblock \emph{Phys. Rev. A}, 69:\penalty0 052330, May 2004.
\newblock \doi{10.1103/PhysRevA.69.052330}.
\newblock URL \url{https://link.aps.org/doi/10.1103/PhysRevA.69.052330}.

\bibitem[Calderbank et~al.(1998)Calderbank, Rains, Shor, and Sloane]{intro-QEC-3}
A.R. Calderbank, E.M. Rains, P.M. Shor, and N.J.A. Sloane.
\newblock Quantum error correction via codes over gf(4).
\newblock \emph{IEEE Transactions on Information Theory}, 44\penalty0 (4):\penalty0 1369--1387, 1998.
\newblock \doi{10.1109/18.681315}.

\bibitem[Gottesman(2009)]{QECCgottesman2009introductionquantumerrorcorrection}
Daniel Gottesman.
\newblock An introduction to quantum error correction and fault-tolerant quantum computation, 2009.
\newblock URL \url{https://arxiv.org/abs/0904.2557}.

\bibitem[Cirac et~al.(1997)Cirac, Zoller, Kimble, and Mabuchi]{intro-QN-1}
J.~I. Cirac, P.~Zoller, H.~J. Kimble, and H.~Mabuchi.
\newblock Quantum state transfer and entanglement distribution among distant nodes in a quantum network.
\newblock \emph{Phys. Rev. Lett.}, 78:\penalty0 3221--3224, Apr 1997.
\newblock \doi{10.1103/PhysRevLett.78.3221}.
\newblock URL \url{https://link.aps.org/doi/10.1103/PhysRevLett.78.3221}.

\bibitem[Hillery et~al.(1999)Hillery, Bu\ifmmode~\check{z}\else \v{z}\fi{}ek, and Berthiaume]{intro-secret-sharing}
Mark Hillery, Vladim\'{\i}r Bu\ifmmode~\check{z}\else \v{z}\fi{}ek, and Andr\'e Berthiaume.
\newblock Quantum secret sharing.
\newblock \emph{Phys. Rev. A}, 59:\penalty0 1829--1834, Mar 1999.
\newblock \doi{10.1103/PhysRevA.59.1829}.
\newblock URL \url{https://link.aps.org/doi/10.1103/PhysRevA.59.1829}.

\bibitem[Hein et~al.(2006)Hein, Dür, Eisert, Raussendorf, den Nest, and Briegel]{gstates_review}
M.~Hein, W.~Dür, J.~Eisert, R.~Raussendorf, M.~Van den Nest, and H.~J. Briegel.
\newblock Entanglement in graph states and its applications, 2006.

\bibitem[Hein et~al.(2004)Hein, Eisert, and Briegel]{Eisert}
M.~Hein, J.~Eisert, and H.~J. Briegel.
\newblock Multiparty entanglement in graph states.
\newblock \emph{Physical Review A}, 69\penalty0 (6), June 2004.
\newblock ISSN 1094-1622.
\newblock \doi{10.1103/physreva.69.062311}.
\newblock URL \url{http://dx.doi.org/10.1103/PhysRevA.69.062311}.

\bibitem[Raussendorf et~al.(2002)Raussendorf, Browne, and Briegel]{mb1Raussendorf_2002}
Robert Raussendorf, Daniel Browne, and Hans Briegel.
\newblock The one-way quantum computer--a non-network model of quantum computation.
\newblock \emph{Journal of Modern Optics}, 49\penalty0 (8):\penalty0 1299–1306, July 2002.
\newblock ISSN 1362-3044.
\newblock \doi{10.1080/09500340110107487}.
\newblock URL \url{http://dx.doi.org/10.1080/09500340110107487}.

\bibitem[Raissi et~al.(2020)Raissi, Teixidó, Gogolin, and Acín]{Raissi_2020}
Zahra Raissi, Adam Teixidó, Christian Gogolin, and Antonio Acín.
\newblock Constructions of $k$-uniform and absolutely maximally entangled states beyond maximum distance codes.
\newblock \emph{Physical Review Research}, 2\penalty0 (3), September 2020.
\newblock ISSN 2643-1564.
\newblock \doi{10.1103/physrevresearch.2.033411}.
\newblock URL \url{http://dx.doi.org/10.1103/PhysRevResearch.2.033411}.

\bibitem[Raissi et~al.(2018)Raissi, Gogolin, Riera, and Acín]{Raissi_2018}
Zahra Raissi, Christian Gogolin, Arnau Riera, and Antonio Acín.
\newblock Optimal quantum error correcting codes from absolutely maximally entangled states.
\newblock \emph{Journal of Physics A: Mathematical and Theoretical}, 51\penalty0 (7):\penalty0 075301, January 2018.
\newblock ISSN 1751-8121.
\newblock \doi{10.1088/1751-8121/aaa151}.
\newblock URL \url{http://dx.doi.org/10.1088/1751-8121/aaa151}.

\bibitem[Schlingemann(2001)]{eisert_schlingemann2001}
D.~Schlingemann.
\newblock Stabilizer codes can be realized as graph codes, 2001.
\newblock URL \url{https://arxiv.org/abs/quant-ph/0111080}.

\bibitem[Cohen et~al.(2005)Cohen, Honkala, Litsyn, and Lobstein]{colorabilitydef}
G{\'e}rard~D. Cohen, Iiro~S. Honkala, Simon Litsyn, and Antoine Lobstein.
\newblock Covering codes.
\newblock In \emph{North-Holland Mathematical Library}, 2005.
\newblock URL \url{https://api.semanticscholar.org/CorpusID:195891379}.

\bibitem[Severini(2006)]{Severini_2006}
Simone Severini.
\newblock Two-colorable graph states with maximal schmidt measure.
\newblock \emph{Physics Letters A}, 356\penalty0 (2):\penalty0 99–103, July 2006.
\newblock ISSN 0375-9601.
\newblock \doi{10.1016/j.physleta.2006.03.026}.
\newblock URL \url{http://dx.doi.org/10.1016/j.physleta.2006.03.026}.

\bibitem[Hedayat et~al.(1999)Hedayat, Sloane, and Stufken]{OAdefhedayat1999orthogonal}
A.~S. Hedayat, N.~J.~A. Sloane, and J.~Stufken.
\newblock \emph{Orthogonal Arrays: Theory and Applications}.
\newblock Springer-Verlag, New York, 1999.

\bibitem[Goyeneche and Życzkowski(2014)]{QOA-def1}
Dardo Goyeneche and Karol Życzkowski.
\newblock Genuinely multipartite entangled states and orthogonal arrays.
\newblock \emph{Physical Review A}, 90\penalty0 (2), August 2014.
\newblock ISSN 1094-1622.
\newblock \doi{10.1103/physreva.90.022316}.
\newblock URL \url{http://dx.doi.org/10.1103/PhysRevA.90.022316}.

\bibitem[Seveso et~al.(2017{\natexlab{a}})Seveso, Goyeneche, and Życzkowski]{QOA-def2}
Luigi Seveso, Dardo Goyeneche, and Karol Życzkowski.
\newblock All orthogonal arrays from quantum states.
\newblock 2017{\natexlab{a}}.
\newblock URL \url{https://api.semanticscholar.org/CorpusID:125400057}.

\bibitem[Goyeneche et~al.(2018)Goyeneche, Raissi, Di~Martino, and \ifmmode~\dot{Z}\else \.{Z}\fi{}yczkowski]{zahra_qoa}
Dardo Goyeneche, Zahra Raissi, Sara Di~Martino, and Karol \ifmmode~\dot{Z}\else \.{Z}\fi{}yczkowski.
\newblock Entanglement and quantum combinatorial designs.
\newblock \emph{Phys. Rev. A}, 97:\penalty0 062326, Jun 2018.
\newblock \doi{10.1103/PhysRevA.97.062326}.
\newblock URL \url{https://link.aps.org/doi/10.1103/PhysRevA.97.062326}.

\bibitem[Zha et~al.(2023)Zha, Ahmed, Imran, Rizvi, Magsi, and Zhang]{QOA-def3}
Xin-Wei Zha, Irfan Ahmed, Muhammad Imran, Syed Maaz~Ahmed Rizvi, Hina Magsi, and Yanpeng Zhang.
\newblock Phase parameters of orthogonal arrays and special maximally multi-qubit entangled states.
\newblock \emph{Modern Physics Letters B}, 2023.
\newblock URL \url{https://api.semanticscholar.org/CorpusID:265559569}.

\bibitem[Zang et~al.(2023)Zang, Tian, Fei, and Zuo]{QOA-kUNI-1}
Yajuan Zang, Zihong Tian, S.~Fei, and Huijuan Zuo.
\newblock Quantum k-uniform states from quantum orthogonal arrays.
\newblock \emph{International Journal of Theoretical Physics}, 62:\penalty0 1--22, 2023.

\bibitem[Lin et~al.(2024)Lin, Pang, and Wang]{QOA-applications-1}
Xiao Lin, Shanqi Pang, and Jing Wang.
\newblock Construction of asymmetric orthogonal arrays with high strength.
\newblock \emph{Stat}, 2024.

\bibitem[Rötteler and Wocjan(2004)]{QOA-applications-2}
M.~Rötteler and P.~Wocjan.
\newblock Equivalence of decoupling schemes and orthogonal arrays.
\newblock \emph{IEEE Transactions on Information Theory}, 52:\penalty0 4171--4181, 2004.

\bibitem[Burchardt and Raissi(2020)]{Burchardt_2020}
Adam Burchardt and Zahra Raissi.
\newblock Stochastic local operations with classical communication of absolutely maximally entangled states.
\newblock \emph{Physical Review A}, 102\penalty0 (2), August 2020.
\newblock ISSN 2469-9934.
\newblock \doi{10.1103/physreva.102.022413}.
\newblock URL \url{http://dx.doi.org/10.1103/PhysRevA.102.022413}.

\bibitem[Bouchet(1993)]{BOUCHET199375}
André Bouchet.
\newblock Recognizing locally equivalent graphs.
\newblock \emph{Discrete Mathematics}, 114\penalty0 (1):\penalty0 75--86, 1993.
\newblock ISSN 0012-365X.
\newblock \doi{https://doi.org/10.1016/0012-365X(93)90357-Y}.
\newblock URL \url{https://www.sciencedirect.com/science/article/pii/0012365X9390357Y}.

\bibitem[Van~den Nest et~al.(2004{\natexlab{a}})Van~den Nest, Dehaene, and De~Moor]{Van_den_Nest_2004}
Maarten Van~den Nest, Jeroen Dehaene, and Bart De~Moor.
\newblock Graphical description of the action of local clifford transformations on graph states.
\newblock \emph{Physical Review A}, 69\penalty0 (2), February 2004{\natexlab{a}}.
\newblock ISSN 1094-1622.
\newblock \doi{10.1103/physreva.69.022316}.
\newblock URL \url{http://dx.doi.org/10.1103/PhysRevA.69.022316}.

\bibitem[Van~den Nest et~al.(2004{\natexlab{b}})Van~den Nest, Dehaene, and De~Moor]{Van_den_Nest_2004_alg}
Maarten Van~den Nest, Jeroen Dehaene, and Bart De~Moor.
\newblock Efficient algorithm to recognize the local clifford equivalence of graph states.
\newblock \emph{Physical Review A}, 70\penalty0 (3), September 2004{\natexlab{b}}.
\newblock ISSN 1094-1622.
\newblock \doi{10.1103/physreva.70.034302}.
\newblock URL \url{http://dx.doi.org/10.1103/PhysRevA.70.034302}.

\bibitem[Raissi et~al.(2022)Raissi, Burchardt, and Barnes]{zahra_adam_kuni}
Zahra Raissi, Adam Burchardt, and Edwin Barnes.
\newblock General stabilizer approach for constructing highly entangled graph states.
\newblock \emph{Physical Review A}, 106\penalty0 (6), December 2022.
\newblock ISSN 2469-9934.
\newblock \doi{10.1103/physreva.106.062424}.
\newblock URL \url{http://dx.doi.org/10.1103/PhysRevA.106.062424}.

\bibitem[West(2001)]{eisert_g1}
Douglas~B. West.
\newblock \emph{Introduction to Graph Theory}.
\newblock Prentice Hall, Upper Saddle River, N.J., second edition, 2001.
\newblock ISBN 0130144002 9780130144003.

\bibitem[Diestel(2005)]{eisert_g2}
Reinhard Diestel.
\newblock \emph{Graph Theory (Graduate Texts in Mathematics)}.
\newblock Springer, 2005.
\newblock ISBN 3540261826.

\bibitem[Bahramgiri and Beigi(2007)]{quditsLCbahramgiri2007graphstatesactionlocal}
Mohsen Bahramgiri and Salman Beigi.
\newblock Graph states under the action of local clifford group in non-binary case, 2007.
\newblock URL \url{https://arxiv.org/abs/quant-ph/0610267}.

\bibitem[Raissi et~al.(2024)Raissi, Barnes, and Economou]{zahra_qudits}
Zahra Raissi, Edwin Barnes, and Sophia~E. Economou.
\newblock Deterministic generation of qudit photonic graph states from quantum emitters.
\newblock \emph{PRX Quantum}, 5:\penalty0 020346, May 2024.
\newblock \doi{10.1103/PRXQuantum.5.020346}.
\newblock URL \url{https://link.aps.org/doi/10.1103/PRXQuantum.5.020346}.

\bibitem[D\"ur et~al.(2003)D\"ur, Aschauer, and Briegel]{eisert_5}
W.~D\"ur, H.~Aschauer, and H.-J. Briegel.
\newblock Multiparticle entanglement purification for graph states.
\newblock \emph{Phys. Rev. Lett.}, 91:\penalty0 107903, Sep 2003.
\newblock \doi{10.1103/PhysRevLett.91.107903}.
\newblock URL \url{https://link.aps.org/doi/10.1103/PhysRevLett.91.107903}.

\bibitem[Zander and Becker(2024)]{zander2024benchmark}
René Zander and Colin Kai-Uwe Becker.
\newblock Benchmarking multipartite entanglement generation with graph states, 2024.
\newblock URL \url{https://arxiv.org/abs/2402.00766}.

\bibitem[Eisert and Briegel(2001)]{SM-Eisert}
Jens Eisert and Hans~J. Briegel.
\newblock Schmidt measure as a tool for quantifying multiparticle entanglement.
\newblock \emph{Phys. Rev. A}, 64:\penalty0 022306, Jul 2001.
\newblock \doi{10.1103/PhysRevA.64.022306}.
\newblock URL \url{https://link.aps.org/doi/10.1103/PhysRevA.64.022306}.

\bibitem[Karp(1972)]{Karp1972}
Richard~M. Karp.
\newblock \emph{Reducibility among Combinatorial Problems}, pages 85--103.
\newblock Springer US, Boston, MA, 1972.
\newblock ISBN 978-1-4684-2001-2.
\newblock \doi{10.1007/978-1-4684-2001-2_9}.
\newblock URL \url{https://doi.org/10.1007/978-1-4684-2001-2_9}.

\bibitem[Stockmeyer(1973)]{stockmeyer1973planar}
L.~Stockmeyer.
\newblock Planar 3-colorability is polynomial complete.
\newblock \emph{SIGACT News}, 5\penalty0 (1), 1973.

\bibitem[Lovasz(1973)]{lovasz1973coverings}
L.~Lovasz.
\newblock Coverings and colorings of hypergraphs.
\newblock In \emph{Proceedings of the 4th Southeastern Conference on Combinatorics, Graph Theory, and Computing}, pages 3--12, 1973.
\newblock URL \url{http://www.cs.elte.hu/~lovasz/scans/covercolor.pdf}.

\bibitem[Seveso et~al.(2017{\natexlab{b}})Seveso, Goyeneche, and Życzkowski]{QOA-construction-1}
L.~Seveso, D.~Goyeneche, and K.~Życzkowski.
\newblock Coarse-grained entanglement classification through orthogonal arrays.
\newblock \emph{Journal of Mathematical Physics}, 2017{\natexlab{b}}.

\bibitem[Vandré et~al.(2024)Vandré, de~Jong, Hahn, Burchardt, Gühne, and Pappa]{guhne_marginals}
Lina Vandré, Jarn de~Jong, Frederik Hahn, Adam Burchardt, Otfried Gühne, and Anna Pappa.
\newblock Distinguishing graph states by the properties of their marginals, 2024.
\newblock URL \url{https://arxiv.org/abs/2406.09956}.

\bibitem[Dür et~al.(2000)Dür, Vidal, and Cirac]{GHZ_vs_W}
W.~Dür, G.~Vidal, and J.~I. Cirac.
\newblock Three qubits can be entangled in two inequivalent ways.
\newblock \emph{Physical Review A}, 62\penalty0 (6), November 2000.
\newblock ISSN 1094-1622.
\newblock \doi{10.1103/physreva.62.062314}.
\newblock URL \url{http://dx.doi.org/10.1103/PhysRevA.62.062314}.

\bibitem[Booth and Carette(2022)]{qudit_LC_1}
Robert~I. Booth and Titouan Carette.
\newblock Complete zx-calculi for the stabiliser fragment in odd prime dimensions.
\newblock Schloss Dagstuhl – Leibniz-Zentrum für Informatik, 2022.
\newblock \doi{10.4230/LIPICS.MFCS.2022.24}.
\newblock URL \url{https://drops.dagstuhl.de/entities/document/10.4230/LIPIcs.MFCS.2022.24}.

\bibitem[Booth(2022)]{qudit_LC_2}
Robert~Ivan Booth.
\newblock \emph{{Measurement-based quantum computation beyond qubits}}.
\newblock Theses, {Sorbonne Universit{\'e}}, February 2022.
\newblock URL \url{https://theses.hal.science/tel-03867179}.

\bibitem[Ji et~al.(2008)Ji, Chen, Wei, and Ying]{ji2008lulcconjecturefalse}
Zhengfeng Ji, Jianxin Chen, Zhaohui Wei, and Mingsheng Ying.
\newblock The lu-lc conjecture is false, 2008.
\newblock URL \url{https://arxiv.org/abs/0709.1266}.

\bibitem[Tsimakuridze and Gühne(2017)]{guhne_lclu}
Nikoloz Tsimakuridze and Otfried Gühne.
\newblock Graph states and local unitary transformations beyond local clifford operations.
\newblock \emph{Journal of Physics A: Mathematical and Theoretical}, 50\penalty0 (19):\penalty0 195302, April 2017.
\newblock ISSN 1751-8121.
\newblock \doi{10.1088/1751-8121/aa67cd}.
\newblock URL \url{http://dx.doi.org/10.1088/1751-8121/aa67cd}.

\bibitem[Adcock et~al.(2020)Adcock, Morley-Short, Dahlberg, and Silverstone]{lcrobits}
Jeremy~C. Adcock, Sam Morley-Short, Axel Dahlberg, and Joshua~W. Silverstone.
\newblock Mapping graph state orbits under local complementation.
\newblock \emph{Quantum}, 4:\penalty0 305, August 2020.
\newblock ISSN 2521-327X.
\newblock \doi{10.22331/q-2020-08-07-305}.
\newblock URL \url{http://dx.doi.org/10.22331/q-2020-08-07-305}.

\bibitem[Latorre and Sierra(2015)]{kuni-1}
Jose~I. Latorre and German Sierra.
\newblock Holographic codes, 2015.
\newblock URL \url{https://arxiv.org/abs/1502.06618}.

\bibitem[Raissi(2020)]{kuni-2}
Zahra Raissi.
\newblock {Modifying Method of Constructing Quantum Codes From Highly Entangled States}.
\newblock \emph{IEEE Access}, 8:\penalty0 222439--222448, 2020.
\newblock \doi{10.1109/ACCESS.2020.3043401}.

\bibitem[Pastawski et~al.(2015)Pastawski, Yoshida, Harlow, and Preskill]{kuni-happy-paper}
Fernando Pastawski, Beni Yoshida, Daniel Harlow, and John Preskill.
\newblock Holographic quantum error-correcting codes: toy models for the bulk/boundary correspondence.
\newblock \emph{Journal of High Energy Physics}, 2015\penalty0 (6), June 2015.
\newblock ISSN 1029-8479.
\newblock \doi{10.1007/jhep06(2015)149}.
\newblock URL \url{http://dx.doi.org/10.1007/JHEP06(2015)149}.

\end{thebibliography}

\onecolumngrid

\appendix
\section{Some useful Properties}\label{appen:useful}

Let us present some useful for our analysis relations. The first one is the commutation of $H^\dagger CZ\ket{kj}$, which is given by the following equation:
\begin{equation}
\label{eq:hdagczcomm}
    H^\dagger CZ\ket{kj}=CXH^\dagger\ket{kj},
\end{equation}
and it is proven here. Let us assume a two-body state $\ket{k,j}$ with $k,j\in\mathbb{Z}_d$, where d is the local dimension. 
\begin{align}
    H^\dagger CZ\ket{kj}=\omega^{kj}H^\dagger \ket{kj}=\displaystyle\frac{1}{\sqrt{d}}\omega^{kj}\ket{k}\sum_{l=0}^{d-1}\omega^{-lj}\ket{l}.
\end{align}
On the other hand, let us apply the following assuming the control is happening always on the first particle:
\begin{align}
    CXH^\dagger\ket{kj}&=\displaystyle\frac{1}{\sqrt{d}}CX\ket{k}\sum_{l=0}^{d-1}\omega^{-lj}\ket{l} =\frac{1}{\sqrt{d}}\ket{k}\sum_{l=0}^{d-1}\omega^{-lj}\ket{l+k} =\frac{1}{\sqrt{d}}\ket{k}\sum_{l=0}^{d-1}\omega^{kj}\omega^{-kj}\omega^{-lj}\ket{l+k} \nonumber\\
    &=\frac{1}{\sqrt{d}}\omega^{kj}\ket{k}\sum_{l=0}^{d-1}\omega^{-(l+k)j}\ket{l+k} =\frac{1}{\sqrt{d}}\omega^{kj}\ket{k}\sum_{l^\prime=0}^{d-1}\omega^{-l^\prime j}\ket{l^\prime}. 
\end{align}
Therefore the proof is completed.
Another useful property is that the $X$ operator to the power $a\in\mathbb{F}_d$ is given by the following equation:
\begin{equation}
\centering\label{eq:xacommut}
    X^a = H^{\dagger}Z^aH.
\end{equation}
The above property is proven here. The application of $X^a$ yields $X^a\ket{i} =\ket{i+a}$. While for the $H^{\dagger}Z^aH$ we have:
\begin{align*}
    &H\ket{i} = \sum_{l=0}^{d-1}\omega^{il}\ket{l}\\
    &Z^aH\ket{i} = \sum_{l=0}^{d-1}\omega^{il}\omega^{al}\ket{l}\\
    &H^{\dagger}Z^aH\ket{i} = \sum_{m=0}^{d-1}\sum_{l=0}^{d-1}\omega^{il}\omega^{al}\omega^{-ml}\ket{m},
\end{align*}
which leads to $H^{\dagger}Z^aH\ket{i} =\ket{i+a}$ and therefore the proof is completed. Finally, a mathematical fact we are going to use extensively in our analysis is:
\begin{equation}
\label{eq:delta-prop} 
    \frac{1}{N}\sum_l \left(\text{exp}\left(\frac{2\pi i }{N}(k-j)\right)\right)^l = \delta_{k,j}, \quad \text{for} \quad k,j\in\mathbb{Z},
\end{equation}
where $\delta_{i,j}$ is the Kronecker delta.
\section{Proof of proposition \ref{thm-3col}}\label{app:proof-3col-thm}
\begin{proposition}
Let us assume a three-colorable graph state described by the adjacency matrix (\ref{eq:gamma-3col-studied}), defined as in equation (\ref{eq:3coldef}), satisfying the set conditions described in the step-by-step construction at the beginning of this section and without loss of generality $n_R\leq n_{B_u}$, and $n_G\leq n_{B_c\setminus G}$. Then, any such graph state satisfies the following equation:

\begin{align}
\label{eq:3colfinalresultv2appen}
     H^{\dagger\otimes B_c\setminus G} H^{\dagger\otimes B_u} \ket{\psi_{\text{3-color}}} = \sum_{\vec{u}=0}^{d-1} \ket{\vec{u}\mathcal{G}_{B_uR}}\Delta\sum_{\vec{g}=0}^{d-1}\ket{\vec{g}\mathcal{G}_{B_c\setminus G}},
\end{align}
where $\vec{u}=(u_1,u_2...,u_{n_R})$, $\vec{g}=(g_1,\cdots,g_{n_G})$ both row vectors, $\mathcal{G}_{B_uR}=\begin{bmatrix} \mathbb{I}_{n_R} & | & A_{B_uR}^T\end{bmatrix}$, $\mathcal{G}_{B_c\setminus G}=\begin{bmatrix} \mathbb{I}_{n_R} & | & A_{GB_c\setminus G}^T\end{bmatrix}$, $f_{k\in B_c} = \sum_{r\in R}\Gamma_{rk}u_r$, and the operator $\Delta$ is defined as following:
\begin{equation}
\label{eq:thm-delta-def-appen}
    \Delta = \left(\bigotimes_{i=1}^{n_G}Z_{g_i}^{f_{g_i}}\right)\left(\bigotimes_{j=1}^{n_{B_c\setminus G}}X_{b_j}^{f_{b_j}}\right).
\end{equation}
\end{proposition}

\begin{proof}
Following the steps of the proof for the closed-form of the $\ket{\psi_{\text{2-color}}}$, $H^{\dagger\otimes B} = H^{\dagger\otimes B_c}H^{\dagger\otimes B_u}$ is applied:
\begin{equation}
    \label{eq:3colgenproof1}
    H^{\dagger\otimes B_c} H^{\dagger\otimes B_u} \ket{\psi_{\text{3-color}}} = H^{\dagger\otimes B_c} \left( \prod_{b, b' \in B_c} C_b Z_{b'}^{\Gamma_{bb'}} \right) \left( \prod_{b\in B_c, r\in R} C_r Z_b^{\Gamma_{br}} \right)H^{\dagger\otimes B_u} \left( \prod_{b\in B_u, r\in R} C_r Z_b^{\Gamma_{br}} \right) \ket{+}^{\otimes R} \ket{+}^{\otimes B_u} \ket{+}^{\otimes B_c}.
\end{equation}
It is clear that if $n_R\leq n_{B_u}$ and $\vec{u}=(u_1,u_2...,u_{n_R})$, a row vector :
\begin{equation}
    H^{\dagger\otimes B_u} \left( \prod_{b\in B_u, r\in R} C_b Z_{r}^{\Gamma_{br}} \right) \ket{+}^{\otimes R} \ket{+}^{\otimes B_u} = \displaystyle\sum_{\vec{u}=0}\ket{\vec{u}}\bigotimes_{b\in B_u}\ket{\sum_{r\in R}\Gamma_{rb}u_r}.
\end{equation}
Therefore, it is possible to obtain the following:
\begin{equation}
\label{eq:3colproofafterbu}
    H^{\dagger\otimes B_c} H^{\dagger\otimes B_u} \ket{\psi_{\text{3-color}}} = \displaystyle\sum_{\vec{u}=0} H^{\dagger\otimes B_c} \left( \prod_{b, b' \in B_c} C_b Z_{b'}^{\Gamma_{bb'}} \right) \left( \prod_{b\in B_c, r\in R} C_b Z_{r}^{\Gamma_{br}} \right)\ket{\vec{u}}\bigotimes_{b\in B_u}\ket{\sum_{r\in R}\Gamma_{rb}u_r} H^{\otimes B_c}\ket{0}^{\otimes B_c}.
\end{equation}
With consideration of the block $A_{B_uR}$ of the matrix (\ref{eq:gamma-3col-studied}) the above equation can be written as:
\begin{equation}
\label{eq:3colproofafterbushort}
    H^{\dagger\otimes B_c} H^{\dagger\otimes B_u} \ket{\psi_{\text{3-color}}} = \displaystyle\sum_{\vec{u}=0} H^{\dagger\otimes B_c} \left( \prod_{b, b' \in B_c} C_b Z_{b'}^{\Gamma_{bb'}} \right) \left( \prod_{b\in B_c, r\in R} C_r Z_{b}^{\Gamma_{br}} \right)\ket{\vec{u}\mathcal{G}_{B_uR}} H^{\otimes B_c}\ket{0}^{\otimes B_c}.
\end{equation}
where $\mathcal{G}_{B_uR}=\begin{bmatrix} \mathbb{I}_{n_R} & | & A_{B_uR}^T\end{bmatrix}$. However, this form is not helpful to proceed with our calculation, therefore we are going to expand $\ket{\vec{u}\mathcal{G}_{B_uR}^T}$ and use the closed-form expression at the ending point of our proof. The primary goal is to commute $ H^{\dagger\otimes B_c}$ through the two products with the controlled-Z operations. For this reason, the following is performed:
\begin{equation}
    \label{eq:example}
    H^{\dagger\otimes B_c}\left( \prod_{b, b' \in B_c} C_b Z_{b'}^{\Gamma_{bb'}} \right) =H^{\dagger\otimes B_c}\left( \prod_{b, b' \in B_c} C_b Z_{b'}^{\Gamma_{bb'}} \right)H^{\otimes B_c}H^{\dagger\otimes B_c}=\mathcal{O}_{b,b'\in B_c}H^{\dagger\otimes B_c}.
\end{equation}
For clarification, we have defined an operator named $\mathcal{O}_{b,b'\in B_c}$ as follows:
\begin{equation}
\label{eq:o-operator-definition}
    \mathcal{O}_{b,b'\in B_c} = H^{\dagger\otimes B_c}\left( \prod_{b, b' \in B_c} C_b Z_{b'}^{\Gamma_{bb'}} \right)H^{\otimes B_c}.
\end{equation}
The operator defined in equation (\ref{eq:o-operator-definition}) will help us commutes as following:
\begin{equation}
    H^{\dagger\otimes B_c}\left( \prod_{b\in B_c, r\in R} C_b Z_{r}^{\Gamma_{br}} \right) = \left( \prod_{b\in B_c, r\in R} C_b X_{r}^{\Gamma_{br}} \right) H^{\dagger\otimes B_c},
\end{equation}
and thus equation (\ref{eq:3colproofafterbushort}) can be written as:
\begin{equation}
    H^{\dagger\otimes B_c} H^{\dagger\otimes B_u} \ket{\psi_{\text{3-color}}} = \sum_{\vec{u}=0}^{d-1} \mathcal{O}_{b,b'\in B_c} \left( \prod_{r\in R, b\in B_c} C_r X_b^{\Gamma_{rb}} \right) \ket{\vec{u}}\left( \bigotimes_{b\in B_u} \ket{\sum_{r\in R} \Gamma_{rb} u_r} \right) \ket{0}^{\otimes B_c}. 
\end{equation}
Finally, using similar steps as the concluding steps in the proof of proposition \ref{thm-2col} and recalling that $n_R\leq n_{B_c}$:
\begin{equation}
    H^{\dagger\otimes B_c} H^{\dagger\otimes B_u} \ket{\psi_{\text{3-color}}} = \sum_{\vec{u}=0}^{d-1} \left(\ket{\vec{u}}  \bigotimes_{b\in B_u} \ket{\sum_{r\in R} \Gamma_{rb} u_r} \right) \left( \mathcal{O}_{b,b'\in B_c} \bigotimes_{b\in B_c} \ket{\sum_{r\in R} \Gamma_{rb} u_r} \right).
\end{equation}
Therefore we can write the following:
\begin{equation}
\label{eq:appen:proof-1}
    H^{\dagger\otimes B_c} H^{\dagger\otimes B_u} \ket{\psi_{\text{3-color}}} = \sum_{\vec{u}=0}^{d-1} \ket{\vec{u}\mathcal{G}_{B_uR}} \left( \mathcal{O}_{b,b'\in B_c} \bigotimes_{b\in B_c} \ket{\sum_{r\in R} \Gamma_{rb} u_r} \right).
\end{equation}
This is a great point to underline, that the connections of the blue particles imposed in the setup are apparent in the above equation, as well and the fact that the particles belonging to the $B_c$ set are still connected with the red particles is obvious from the coupling in the $\ket{\sum_{r\in R} \Gamma_{rb} u_r}$. Hence, we are going to define the parameter $f_k$ as follows:
\begin{equation}
\label{eq:appen_fk-def}
    f_k=\sum_{r\in R}\Gamma_{rk}u_r.
\end{equation}
This parameter describes the connections of a specific particle named $k$ with every red particle. By construction, the particles belonging in the $B_c$ are assumed to be a two-colorable subgraph, and thus the total graph is a three-colorable one. This means that the $B_c$ graph has particles that maintain their blue color. These particles belong to a set which we are going to denote as $B_c\setminus G$ and this is a set with elements $B_c\setminus G= \{b'_1,\cdots,b'_{n_{B_c\setminus G}}\}$. This means that the total number of particles belonging to the $B_c$ graph that maintain their blue color is $n_{B_c\setminus G}$. Following the same structure, the $B_c$ graph has some particles that became green. Their total number is $n_G$, they belong to the set $G$ and this set has elements denoted as $G= \{g_1,\cdots,g_{n_{G}}\}$. With these points in mind, we have that: 
\begin{equation}
    \bigotimes_{b\in B_c}\ket{\sum_{r\in R}\Gamma_{rb}u_r} = \bigotimes_{k\in B_c}\ket{f_k} = \bigotimes_{g\in G}\ket{f_g}\bigotimes_{b\in B_c\setminus G}\ket{f_{b}}.
\end{equation}
Bringing us back to the equation (\ref{eq:appen:proof-1}) we have to calculate the following:
\begin{equation}
    \mathcal{O}_{b,b'\in B_c} \bigotimes_{b\in B_c} \ket{\sum_{r\in R} \Gamma_{rb} u_r},
\end{equation}
and we will do it in stages. Until this point the cardinality of $R$, $G$, $B_u$, $B_c$, and $B_c\setminus G$ was denoted as $n_R$, $n_G$, $n_{B_u}$, $n_{B_c}$, and $n_{B_c\setminus G}$ respectively. From this point in the equations of our proof we are going to denote it as $\lvert R \rvert$, $\lvert G \rvert$, $\lvert B_u \rvert$, $\lvert B_c \rvert$, and $\left\lvert B_c\setminus G \right\rvert$. Let us start the calculation by applying the Hadamard gates: 
\begin{equation}
    H^{\otimes G}H^{\otimes B_c\setminus G} \bigotimes_{g\in G}\ket{f_g}\bigotimes_{b\in B_c\setminus G}\ket{f_{b}}=\sum_{l_{g_1},\cdots,l_{g_{\lvert G\rvert}}=0}^{\lvert G\rvert}\sum_{l_{b'_1},\cdots,l_{b'_{\lvert B_c\setminus G\rvert}}=0}^{\lvert B_c\setminus G\rvert} \omega^{\sum_{i=1}^{\lvert G\rvert}l_{g_i}f_{g_i}+\sum_{i=1}^{{\lvert B_c\setminus G \rvert}}l_{b'_i}f_{b'_i}} \ket{l_{g_1},\cdots,l_{g_{\lvert G\rvert}}}\ket{l_{b'_1},\cdots,l_{b'_{\lvert B_c\setminus G\rvert}}}.
\end{equation}
To introduce a short-hand notation, instead of $\sum_{l_{g_1},\cdots,l_{g_{\lvert G \rvert}}=0}^{\lvert G \rvert}$ we are going to write $\sum_{l_g}$. Similarly instead $\sum_{l_{b'_1},\cdots,l_{b'_{\lvert B_c\setminus G \rvert}}=0}^{\lvert B_c\setminus G\rvert}$ we will write $\sum_{l_b}$. After these clarifications, we have to proceed with the application of the corresponding controlled-Z operations. For this, we have to keep in mind the following:
\begin{equation}
    \prod_{b,b'\in B_c}CZ^{\Gamma_{bb'}} = \prod_{g\in G, b \in B_c\setminus G}CZ^{\Gamma_{gb}}.
\end{equation}
 Therefore we have that:
\begin{equation}
    \begin{aligned}
        & \left(\prod_{g\in G, b \in B_c\setminus G}CZ^{\Gamma_{gb}}\right)H^{\otimes G}H^{\otimes B_c\setminus G} \bigotimes_{g\in G}\ket{f_g}\bigotimes_{b\in B_c\setminus G}\ket{f_{b}} \\
        &= \sum_{l_g}\sum_{l_b} \Omega\left(\sum_{i=1}^{\lvert G \rvert}l_{g_i}f_{g_i}+\sum_{i=1}^{\lvert B_c\setminus G\rvert }l_{b'_i}f_{b'_i} + \sum_{i=1}^{\lvert G \rvert }\sum_{j=1}^{\lvert B_c\setminus G\rvert }l_{g_{i}}l_{b_{j}}\Gamma_{g_ib'_j}\right)
        \ket{l_{g_1},\cdots,l_{g_{\lvert G \rvert }}}\ket{l_{b'_1},\cdots,l_{b'_{\lvert B_c\setminus G\lvert }}},
    \end{aligned}
\end{equation}
where $\Omega(x)= \omega^x$ for $x\in\mathbb{R}$. At this stage, the fact that the subgraph $B_c$ is two-colorable and that without loss of generality it is assumed that $\lvert G \lvert \leq \rvert B_c\setminus G\rvert $ is crucial. In the same way, as we did for the two colorable cases, instead of applying $H^{\dagger\otimes B_c\setminus G}H^{\dagger\otimes G}$, we apply only the $H^{\dagger\otimes B_c\setminus G}$. This can be done trivially since $H^{\otimes G}H^{\dagger\otimes B_c\setminus G}H^{\dagger\otimes G} = H^{\dagger\otimes B_c\setminus G}$. Therefore we have that:
\begin{align*}
 &H^{\dagger\otimes B_c\setminus G}\left(\prod_{g \in G, b \in B_c\setminus G} CZ^{\Gamma_{gb}}\right) H^{\otimes G} H^{\otimes B_c\setminus G} 
    \bigotimes_{g \in G} \ket{f_g} \bigotimes_{b \in B_c\setminus G} \ket{f_b} \\
    &= \sum_{l_g} \sum_{l_b} \sum_{m_b}
    \Omega\left(\sum_{i=1}^{\lvert G \rvert } l_{g_i} f_{g_i} + \sum_{i=1}^{\lvert B_c\setminus G \rvert } l_{b'_i} f_{b'_i}+\sum_{i=1}^{\lvert G \rvert } \sum_{j=1}^{\lvert B_c\setminus G \rvert } l_{g_i} l_{b'_j} \Gamma_{g_i b'_j} -\sum_{i=1}^{\lvert B_c\setminus G\rvert } l_{b'_i} m_{b'_i} \right)
    \ket{l_{g_1}, \cdots, l_{g_{\lvert G \rvert }}} \ket{m_{b'_1}, \cdots, m_{b'_{\lvert B_c\setminus G \rvert }}},
\end{align*}
where the summation $\sum_{m_b}$ denotes the following summation: $\sum_{m_{b'_1},\cdots,m_{b'_{\lvert B_c\setminus G \rvert }}=0}^{\lvert B_c\setminus G \rvert }$. Our goal is to use the property of the Kronecker delta that is given in equation (\ref{eq:delta-prop}) using as summation the sums of the $l_b$ indices since they are not a part of the ket in the above equation. For this reason, we have to define the quantity $p_a$ as:
\begin{equation}
    \label{eq:pdef}
    p_a = f_{b'_a}+\sum_{i=1}^{\lvert G \rvert }l_{g_i}\Gamma_{b'_ag_i},
\end{equation}
which underlines the structure of the theorized graph. Initially, the letter $a$ in subscript denotes a blue particle indicated with the index $a$ following the philosophy of the equation (\ref{eq:appen_fk-def}). Subsequently, the second term in the equation (\ref{eq:pdef}) represents the additional connections we imposed to form the $B_c$ graph and therefore represents the connections between the blue particles belonging to the set $B_c\setminus G$ and the green particles which belong to the set $G$. Thus, taking into consideration equation (\ref{eq:pdef}):
\begin{equation}
    \begin{aligned}
       & H^{\dagger\otimes B_c\setminus G}\left(\prod_{g \in G, b \in B_c\setminus G} CZ^{\Gamma_{gb}}\right) H^{\otimes G} H^{\otimes B_c\setminus G} 
        \bigotimes_{g \in G} \ket{f_g} \bigotimes_{b \in B_c\setminus G} \ket{f_b} \\
        &= \sum_{l_g} \sum_{m_b} \delta_{p_1,m_1}\cdots\delta_{p_{\lvert B_c\setminus G\rvert},m_{\lvert B_c\setminus G\rvert}}\times\Omega\left(\sum_{i=1}^{\lvert G\rvert}l_{g_i}f_{g_i}\right) \ket{l_{g_1},\cdots,l_{g_{\lvert G \rvert}}}\ket{m_{b_1},\cdots,m_{b_{\lvert B_c\setminus G\rvert}}},
    \end{aligned}
\end{equation}
and therefore we obtain that:
\begin{equation}
    \label{eq:3colcf1}
    H^{\dagger\otimes B_c\setminus G} \left(\prod_{g \in G, b \in B_c\setminus G} CZ^{\Gamma_{gb}}\right) H^{\otimes G} H^{\otimes B_c\setminus G} 
    \bigotimes_{g \in G} \ket{f_g} \bigotimes_{b \in B_c\setminus G} \ket{f_b}= \sum_{l_g} \Omega\left(\sum_{i=1}^{\lvert G \rvert} l_{g_i} f_{g_i}\right)\ket{l_{g_1}, \cdots, l_{g_{\lvert G \rvert}}} \ket{p_{g_1}, \cdots, p_{g_{\lvert G \lvert}}}.
\end{equation}
A careful inspection of the equation (\ref{eq:3colcf1}) leads us to write it as:
\begin{equation}
    \label{eq:3colcf1-2}
    \begin{aligned}
        & H^{\dagger\otimes B_c\setminus G} \left(\prod_{g \in G, b \in B_c\setminus G} CZ^{\Gamma_{gb}}\right) H^{\otimes G} H^{\otimes B_c\setminus G} 
        \bigotimes_{g \in G} \ket{f_g} \bigotimes_{b \in B_c\setminus G} \ket{f_b} \\
        &= \left(\bigotimes_{i=1}^{n_G}Z_{g_i}^{f_{g_i}}\right)\left(\bigotimes_{j=1}^{n_{B_c\setminus G}}X_{b_j}^{f_{b_j}}\right)\sum_{l_g} \ket{l_{g_1}, \cdots, l_{g_{n_G}}} \bigotimes_{m=1}^{n_{B_c\setminus G}}\ket{\sum_{k=1}^{n_G}l_{g_k}\Gamma_{b_mg_k}}.
    \end{aligned}
\end{equation}
It must be underlined that the notation of cardinality of the sets is switched to the one we adopted at the beginning. To introduce a more effective notation in our equations, let us define the operator $\Delta$:
\begin{equation}
\label{eq:Delta-Def}
    \Delta = \left(\bigotimes_{i=1}^{n_G}Z_{g_i}^{f_{g_i}}\right)\left(\bigotimes_{j=1}^{n_{B_c\setminus G}}X_{b_j}^{f_{b_j}}\right).
\end{equation}
Likewise the proof for proposition \ref{thm-2col}, we will define a row vector $\vec{g}=(g_1,\cdots,g_{n_G})$ and a matrix $\mathcal{G}_{B_c\setminus G}=\begin{bmatrix} \mathbb{I}_{n_R} & | & A_{GB_c\setminus G}^T\end{bmatrix}$, where $A_{GB_c\setminus G}$ is the corresponding block of the matrix (\ref{eq:gamma-3col-studied}). Therefore, the remaining part in the parenthesis of equation (\ref{eq:appen:proof-1}) is calculated and we arrive at:
\begin{equation}
\label{eq:3colfinalresult}
    H^{\dagger\otimes B_c} H^{\dagger\otimes B_u} \ket{\psi_{\text{3-color}}} =\sum_{u_1, \ldots, u_{n_R}=0}^{d-1} \ket{\vec{u}\mathcal{G}_{B_uR}}\Delta\left(\sum_{g_1,\cdots,g_{n_G}=0}^{d-1}\ket{\vec{g}\mathcal{G}_{B_c\setminus G}}\right),
\end{equation}
which concludes our proof.
\end{proof}

\section{One dimensional cluster states and circles as the $B_c$ subgraph}\label{appen:1d-cluster}
The goal of this subsection is to study $B_c$ graphs with specific structures, that can lead to interesting cases of three-colorable graphs. Let us start by assuming that the $B_c$ graph is a one-dimensional cluster state with an arbitrary number of particles, denoted as $k$. To better illustrate it, figure \ref{fig:chain_example_graph} is such an example with $k=9$ particles. We are going to derive our closed-form formula stated in the proposition \ref{thm-3col}, but not starting entirely from the beginning, but rather from the action of the $\mathcal{O}_{b,b'\in B_c}$ on an arbitrary ket $\ket{l_1,\cdots,l_k}$. 
\begin{figure}
    \centering
    \includegraphics[width=0.6\linewidth]{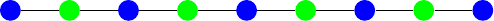}
    \caption{A one-dimensional cluster state with nine vertices as the $B_c$ subgraph.}
    \label{fig:chain_example_graph}
\end{figure}
Let us initiate the analysis by implementing the equation (\ref{eq:o-operator-definition}) on this context:
\begin{align}
    \mathcal{O}_{x, x' \in B_c} = \left( \prod_{x=1}^{k-1} H^{\dagger}_{x_j} \right)\left( \prod_{x=1}^{k-1} C_{x_j}Z_{x_{j+1}}^{\Gamma_{x_jx_{j+1}}} \right)\left( \prod_{x=1}^{k} H_{x_j} \right).
\end{align}
Therefore the action of the above operator on an arbitrary state $\ket{i_1,\cdots,i_k}$ is:

    \begin{equation}
     \left( \prod_{k=1}^{k-1} C_{x_j}Z_{x_{j+1}}^{\Gamma_{x_jx_{j+1}}} \right)\left( \prod_{k=1}^{k} H_{x_j} \right)\ket{i_1,\cdots,i_k}=\sum_{l_1,\cdots,l_k=0}^{d-1}\Omega\left(\sum_{j=1}^{k}l_ji_j+\sum_{j=1}^{k-1}l_jl_{j+1}\Gamma_{x_jx_{j+1}}\right)\ket{l_1,\cdots,l_k}.
    \end{equation}

At this point, it is assumed that the chain contains an odd number of terms. The steps for the even-numbered case are the same and thus are omitted. Therefore, the last particle is denoted as $x_{2j+1}$. Bearing in mind that for the $B_c$ particles $n_G\leq n_{B_c\setminus G}$, the blue particles are always of the form $x_{2j+1}$ (odd-numbered $j$), and the green particles are always of the form $x_{2j}$ (even numbered $j$). Acknowledging the discussion about the number of $H^{\dagger}$ we can apply to decrease the number of terms, only the odd-numbered $H^{\dagger}$ will be performed and therefore the following is obtained. This was explained in Appendix \ref{app:proof-3col-thm}. After a few lines of calculations, we have that:

    \begin{equation}
        \begin{aligned}
            &\left( \prod_{j=0}^{k} H^{\dagger}_{x_{2j+1}} \right)\left( \prod_{j=1}^{k-1} C_{x_j} Z_{x_{j+1}}^{\Gamma_{x_j x_{j+1}}} \right)\left( \prod_{j=1}^{k} H_{x_j} \right) \ket{i_1, \ldots, i_k} \\
            &= \sum_{l_1, \ldots, l_k=0}^{d-1} \sum_{m_1, \ldots, m_{2k+1}=0}^{d-1} \Omega\left(\sum_{j=1}^{k} l_j i_j - \sum_{j=1}^{k} m_{2j-1} l_{2j-1}+\sum_{j=1}^{k-1} l_j l_{j+1} \Gamma_{x_j x_{j+1}}\right) \ket{m_1, l_2, m_3, \ldots, m_{2k+1}}.
        \end{aligned}
    \end{equation}

Taking into account equation (\ref{eq:delta-prop}) we arrive at:
\begin{align}
\label{eq:app:c1}
     &\mathcal{O}_{x, x' \in B_c}\ket{i_1,\cdots,i_k} = \sum_{l_2,\cdots,l_{2k}=0}^{d-1}\Omega\left(\sum_{j=1}^{d-1}l_{2j}i_{2j}\right)\ket{i_1+l_2\Gamma_{x_1x_2},l_2,i_3+l_2\Gamma_{x_3x_2}+l_4\Gamma_{x_3x_4},\cdots},
\end{align}
After some algebraic manipulations, equation (\ref{eq:app:c1}) can be written with respect to the operators $X$ and $Z$:
\begin{align}
\label{eq:operatorchain}
    &\mathcal{O}_{x, x' \in B_c}\ket{i_1,\cdots,i_k} = \bigotimes_{j=0}^{k}X_{x_{2j+1}}^{i_{2j+1}}\bigotimes_{j=0}^{k}Z_{x_{2j}}^{i_{2j}}\sum_{l_2,\cdots,l_{2k}=0}^{d-1}\ket{l_2\Gamma_{x_1x_2},l_2,\cdots,l_{2k}\Gamma_{x_{2k}x_{2k+1}}}.
\end{align}
Equation (\ref{eq:operatorchain}) allows for some intriguing conclusions. If two particles are assumed to belong in the $B_c$ subgraph, then the following is obtained:
\begin{equation}
    \mathcal{O}_{x_1, x_2 \in B_c}\ket{i}_{x_1}\ket{j}_{x_2} =X_{x_1}^{i}\otimes Z_{x_2}^{j}\sum_{l=0}^{d-1}\ket{l}_{x_1}\ket{l}_{x_2},
\end{equation}
which is the Bell basis. To make it more obvious, let us assume we work with qubits and therefore $d=2$:
\begin{equation}
    \mathcal{O}_{x_1, x_2 \in B_c}\ket{i,j} = \ket{i,0} + (-1)^j\ket{i+1,1}.
\end{equation}
This means that the following correspondence is obtained:
\begin{equation}
\begin{aligned}
        \ket{00} &\rightarrow \ket{00} + \ket{11} = \ket{\Phi^{+}} \\
        \ket{01} &\rightarrow \ket{00} - \ket{11} = \ket{\Phi^{-}} \\
        \ket{10} &\rightarrow \ket{10} + \ket{01} = \ket{\Psi^{+}} \\
        \ket{11} &\rightarrow -\ket{10} + \ket{01} = -\ket{\Psi^{-}},
    \end{aligned}
\end{equation}
where $\Phi^{\pm}$ and $\Psi^{\pm}$ are the four Bell states for the qubit case. Similarly, for three particles belonging to the  $B_c$ set, the GHZ state is obtained:
\begin{equation}
    \mathcal{O}_{x_1, x_2, x_3 \in B_c}\ket{i_1,i_2,i_3} = X_{x_1}^{i_1}\otimes Z_{x_2}^{i_2}\otimes X_{x_3}^{i_3}\sum_{l=0}^{d-1}\ket{l}_{x_1}\ket{l}_{x_2}\ket{l}_{x_3}.
\end{equation}
It must be noted that for simplicity in the above example, the values of the adjacency matrix elements are assumed to be one. This is the case as well for the upcoming example where the result for 4 particles is presented:
\begin{align}
    \mathcal{O}_{b, b' \in B_c}\ket{i_1,i_2,i_3,i_4} 
    &= X_{b_1}^{i_1}\otimes Z_{b_2}^{i_2}\otimes X_{b_3}^{i_3}\otimes Z_{b_4}^{i_4} \sum_{l_2,l_4=0}^{d-1}\ket{l_2,l_2,l_2+l_4,l_4}.
\end{align}
This is a  rather interesting result, especially if we realize that we can manipulate the previous equation and arrive at:
\begin{align*}
&H_{b_3}H_{b_4}^{\dagger}\mathcal{O}_{b, b' \in B_c}\ket{i_1,i_2,i_3,i_4}=X_{b_1}^{i_1}\otimes Z_{b_2}^{i_2}\otimes Z_{b_3}^{i_3}\otimes X_{b_b}^{i_4}\sum_{l,m=0}^{d-1}\omega^{lm}\ket{l,l}\otimes\ket{m,m}.
\end{align*}
This result can be generalized for any number of particles and implies two facts. The first is, that indeed the number of $H^{\dagger}$ operations one should apply for the final step of the calculation of the operator defined in general in equation (\ref{eq:o-operator-definition}) remains $H^{\dagger\otimes B_c\setminus G}$. The second fact is, that it is possible to write part, if not the whole ket inside the \ref{eq:operatorchain} as a tensor product having combinations of $\ket{l,l}$ or $\ket{l,l,l}$. This is rather important for the discussion of the equivalence using local operators between the state $\ket{\psi_{\text{2-color}}}$ and the state $\ket{\psi_{\text{3-color}}}$.

At this point, it has to be noted that in the given examples the value of the adjacency matrix is assumed to be one. We recall that every possible value inside the kets belongs to $\mathbb{F}_d$ which is also true for every possible value of the adjacency matrix. It is a well-known mathematical result, that the multiplication of any element belonging to $\mathbb{F}_d$ with any other element belonging to the same field maps to an element in $\mathbb{F}_d$. To grasp this idea, one can consider the simplest qudit case, the qutrit ($d=3$). For this, assume that every value is $\Gamma = 2$. The possible values inside the brackets are 0, 1, or 2. By multiplying every element with 2, we get, by keeping the order, 0,2, and 1, and since the order does not matter, we have every number only once. This allows us to understand that showing that Bell or GHZ state is obtained assuming every adjacency matrix to be 1 can be generalized for any combination of values for the adjacency matrices.

Let us proceed with the $B_c$ particles creating a circle with an even number of particles to create a two-colorable subgraph. Equipped with the previously presented results, it is easy to conclude that if the $B_c$ graph is an even-numbered circle then the action of the operator  $\mathcal{O}_{b,b'\in B_c}$ is given by the following equation:
\begin{align*}
     &\mathcal{O}_{b, b' \in B_c}\ket{c_1,\cdots,c_{2N}} = \bigotimes_{j=0}^{k}X_{b_{2j+1}}^{c_{2j+1}}\bigotimes_{j=0}^{k}Z_{b_{2j}}^{c_{2j}}\sum_{l_2,\cdots,l_{2N}=0}^{d-1}\ket{l_2\Gamma_{b_1b2}+l_{2N}\Gamma_{b_1b_{2N}},l_2,\cdots}.
\end{align*}
This result is connected with the obtained in the equation (\ref{eq:operatorchain}), which is expected, since the only difference between the first and the second setup is that the first and the last particles are connected in the latter case.

\end{document}